\documentclass[12pt, draftclsnofoot, onecolumn]{IEEEtran}

\usepackage{amssymb}
\usepackage{letltxmacro}
\usepackage{pseudocode}
\usepackage{pifont}
\usepackage[]{algorithm2e}
\usepackage[cmex10]{amsmath}
\usepackage[]{algorithm2e}
\usepackage{algorithmic}
\usepackage{color}
\usepackage{cases}
\usepackage{bm}

\usepackage{amsthm}

\ifCLASSINFOpdf
  \usepackage[pdftex]{graphicx}
\usepackage{epstopdf}
\else
  \usepackage[dvips]{graphicx}
\fi

\ifCLASSOPTIONcompsoc

  \usepackage[nocompress]{cite}
\else

  \usepackage{cite}
\fi

\usepackage[caption = false]{subfig}
\usepackage{multirow}

\newtheorem{theorem}{Theorem}[section]

\newtheorem{definition}[theorem]{definition}

\begin{document}

\title{Throughput Analysis of IEEE 802.11 Multi-hop Wireless Networks with Routing Consideration: A General Framework}

\author{Shahbaz~Rezaei,~\IEEEmembership{Professional Member,~IEEE,}
		Mohammed~Gharib,~\IEEEmembership{Professional Member,~IEEE,}        
        and~Ali~Movaghar,~\IEEEmembership{Senior~Member,~IEEE}
\IEEEcompsocitemizethanks{\IEEEcompsocthanksitem S. Rezaei and A. Movaghar are with the Department of Computer Engineering, Sharif University of Technology, Tehran, Iran.  \protect\\
E-mail: shrezaei@ce.sharif.edu; movaghar@sharif.edu   
\IEEEcompsocthanksitem M. Gharib is with the School of Computer Science, Institute for Research in Fundamental Sciences (IPM), Tehran, Iran.\protect\\
E-mail: gharib@ipm.ir}
\thanks{}}

\markboth{This work has been submitted to the IEEE Transactions on Communications (TCOM) for possible publication.}%
{S. Rezaei  \MakeLowercase{\textit{et al.}}: Throughput Analysis of IEEE 802.11 Multi-hop Wireless Networks with Routing Consideration: A General Framework}

\IEEEtitleabstractindextext{
\begin{abstract}
The end-to-end throughput of multi-hop communication in wireless ad hoc networks
is affected by the conflict between
forwarding nodes. It has been shown that sending more 
packets than maximum achievable end-to-end throughput not only  
fails to increase throughput, but also decreases throughput owing 
to high contention and collision. Accordingly, it is of crucial importance 
for a source node to know the maximum end-to-end throughput.  
The end-to-end throughput depends on multiple factors, such as physical 
layer limitations, MAC protocol properties, routing policy and nodes' 
distribution. There have been many studies on analytical modeling of 
end-to-end throughput but none of them has taken routing policy and 
nodes' distribution as well as MAC layer altogether into account. 
In this paper, the end-to-end throughput 
with perfect MAC layer is obtained based on routing policy and nodes' 
distribution in one and two dimensional networks. Then, imperfections of IEEE 
$802.11$ protocol is added to the model to obtain precise value. 
An exhaustive simulation is also made to validate the proposed models 
using NS$2$ simulator. Results show that if the distribution to the next hop for a 
particular routing policy is known, our methodology can obtain the 
maximum end-to-end throughput precisely.

\end{abstract}

}

\maketitle

\IEEEdisplaynontitleabstractindextext

\IEEEpeerreviewmaketitle




\IEEEraisesectionheading{\section{Introduction}\label{sec:introduction}}

\IEEEPARstart{W}{ireless} ad hoc networks allow several nodes located outside 
of the transmission range of one another to communicate through 
intermediate nodes. Using a routing protocol, a source node can 
find a path toward the destination. Nodes selected by the routing 
protocol are responsible for forwarding the source's packets 
until they reach the destination. In a case that source and 
destination are in the transmission range of each other, no 
routing is required and consequently, maximum end-to-end 
throughput is obtained by analyzing physical and Medium 
Access Control (MAC) layer properties. 


In multi-hop wireless communications, achievable throughput is 
significantly lower than single-hop communication due to the 
inevitable transmission overlap of consecutive forwarding nodes. 
Additionally, sending more packets than maximum achievable 
throughput degrades the end-to-end throughput even further. 
Hence, the knowledge about the maximum achievable throughput, 
which can be exploited by source nodes, considerably improves multi-hop 
communications.

In the previous research \cite{ref4,ref5,ref6,ref8,ref9,ref10,ref14,ref17}, 
the analytical expression for maximum throughput was obtained for 
different scenarios. Nevertheless, all of these papers have assumed 
that all nodes placed equi-distance apart in a straight line. In practice, 
wireless nodes are distributed randomly and it is highly unlikely to end 
up with this assumption. Additionally, even if all nodes do follow such a 
restriction, the routing protocol may choose nodes in a way that the 
assumption fails. Thus, it would not be a precise analysis if the end-to-end 
throughput of multi-hop communication is obtained without the routing 
policy and node distribution. 

In this paper, node distribution and routing policy as well as MAC layer 
limitations are taken into consideration to obtain the precise value of end-to-end 
throughput in multi-hop communication  when there is a single flow in the network. 
The geometry of randomly deployed wireless nodes is 
commonly modeled by a Poisson Point Process (PPP), since it is equivalent to placing each node uniformly in an $n$-dimensional space.
\cite{rahmatollahi2012closed,beyme2014stochastic}. Given a $1$-dimensional ($1$-D) PPP,
first, we obtain the maximum end-to-end throughput with perfect MAC and physical layer for two 
different routing policies including random neighbor routing and furthest neighbor 
routing. In perfect MAC and physical layer, it is assumed that a node has no 
negative effect on communications happening outside its interference range. 
Regardless of the existence of such physical layer and MAC protocol, this 
assumption avails us to capture the effect of node distribution and routing 
policy alone on the maximum end-to-end throughput.

Second, we analyze the effect of routing policy and node distribution on 
IEEE $802.11$ MAC protocol and extend our analysis to capture the effects 
of MAC layer limitations. Although the formulas we obtained here are 
confirmed by simulation and are useful for theoretical analysis, they are 
too complicated to be used in wireless nodes. Therefore, we provide a 
simple way to approximate these formulas so that a source node can 
use them for flow and admission control. 

Third, we obtain some promising approximation formulas for maximum 
end-to-end throughput in $2$-D networks which is validated by simulations. 
Our methodology can be used for any routing policy and node distribution 
as long as 
the first and second moments of the distribution of distance between two consecutive 
nodes in a path is known.
To the best of our knowledge, it is the first analytical results about 
maximum throughput in $2$-D multi-hop wireless networks when IEEE 802.11 is used alongside furthest and random neighbor routing.

One direct application of our results resides within the domain of protocol design. For instance, protocols, such as TCP, have inherent flow and congestion control. However, their performance is not near optimum as a result of increasing window size more than optimum value \cite{ref21}. Knowing the maximum achievable throughput can be directly used in protocol design. It can also be used for layer 2 protocol design. Given that the maximum throughput of multi-hop communication is less than single-hop communication, the maximum multi-hop throughput can be used to design layer 2 protocols that avoid greedily sending at the rate of single-hop throughput which contributes to congestion and consequently degrade multi-hop communication throughput.

The remainder of this paper is organized as follows. In section \ref{relatedWork}, 
we briefly present related works in this area. Section \ref{routingAnalysis} obtains 
end-to-end throughput with perfect MAC and physical layer in $1$-D networks. Section \ref{macLayer} 
considers the limitation of lower layers and obtains end-to-end throughput for IEEE 
802.11. Section \ref{approximation} obtains an approximate formula for end-to-end 
throughput which is simple enough to be used in wireless nodes. While it seems 
impossible to obtain an exact formula for maximum throughput in $2$-D networks, 
some approximations are given in Section \ref{$2$-DSection}. A comprehensive 
performance evaluation is carried out in Section \ref{simulation}. Finally, Section 
\ref{conclusion} concludes our work.




\section{Related Work}
\label{relatedWork}

In \cite{nelson1984spatial}, Nelson and Kleinrock obtained the spatial capacity of a wireless network when using slotted Aloha Mac protocol. They assumed that intermediate nodes for routing are selected randomly among nodes located toward a destination. Although they considered both routing and MAC layer in their analysis, slotted Aloha is not commonly used in ad hoc wireless networks. Moreover, it is not clear whether their approach can be generalized to any routing policy.
In \cite{ref7}, Gupta and Kumar derived theoretical bounds for the 
capacity of wireless ad hoc networks. But, in practice, with wireless 
networks based on IEEE $802.11$ MAC protocol, achievable end-to-end 
throughput is far from the derived bounds \cite{ref17}. Taking into 
account the limitation imposed by carrier sensing, authors in \cite{ref20} 
have given an insight over the fundamental discrepancy between single-hop 
and multi-hop communication. They have intuitively asserted that maximum 
end-to-end throughput of multi-hop communication is $n$ times less than 
single-hop throughput where $n$ is the number of consecutive nodes 
whose transmission conflicts with each other. Not only did their simulation 
reveal the incompleteness of their assertion, but also their simulations were 
only conducted for lattice network in which all nodes are placed equi-distance 
apart.

In \cite{ref21}, the effect of MAC layer on TCP performance was analyzed 
and two techniques were proposed to mitigate the performance of TCP. 
However, their analysis only holds when all nodes are placed equi-distance apart. 
In \cite{ref12,ref23}, authors have studied the criteria based on which the 
second or more paths from a source to destination increase throughput. 
However, their analysis only considered the scenario where all nodes are 
placed equi-distance apart and they also assumed that the carrier sensing 
range is equal to transmission range which is not realistic \cite{ref13}.
In \cite{fitzgerald2017analytic}, the throughput of a new backoff scheme for 802.11, called TO-DCF,
was analytically analyzed. Nevertheless, their analysis is only applicable for single-hop networks
in which all nodes are close enough to communicate directly. Hence, no routing mechanism was studied.

Authors in \cite{ref2} proposed a methodology to obtain maximum end-to-end 
throughput in a chain topology by creating contention graph corresponding to 
the topology. However, this method needs the location of all nodes. In \cite{ref11}, 
the deviation based method is proposed to obtain end-to-end throughput 
which uses the central limit theorem to model the throughput deviation. 
However, as it is shown in \cite{ref3}, the deviation based 
model fails to estimate the end-to-end throughput when the number of 
nodes in a route exceeds 4. Authors in \cite{ref3} proposed a method to 
mitigate the deficiencies of the deviation based method which works for 
even long paths. However, this model requires the location of all nodes which 
may not be available. In \cite{tabet2004spatial}, authors obtained the multi-hop throughput in wireless networks when slotted Aloha and incremental redundancy is used. They assumed that nodes are distributed following Poisson point process and two routing policies, namely furthest neighbor routing and nearest neighbor routing, are studied. They showed that coding and retransmissions provide a fully decentralized MAC protocol. However, their analysis cannot be used for commonly use IEEE 802.11 MAC protocol. 

Assuming that all nodes are placed equi-distance apart, the 
authors in \cite{ref1} obtained the maximum end-to-end throughput of chain-topology. Many 
research \cite{ref4,ref5,ref6,ref8,ref9,ref10,ref14,ref17,yin2016throughput} used the same idea 
to obtain a slightly better formula for maximum end-to-end throughput in 
chain-topology. However, none of those works took node distribution and 
routing policy into account. Furthermore, it is not clear how one can extend the models to 2 dimensional networks. In \cite{ref13}, however, the author generalized 
the method used in \cite{ref1} to any chain-topology when the location of 
nodes are known. By running a backward iterative process which is relatively 
complicated for a computationally limited wireless node, a source node can 
obtain the maximum end-to-end throughput. In \cite{yin2016throughput}, the model is extended to support heterogeneous networks
in which nodes transmit different length frames with various loads when all nodes reside in carrier sense range of all other nodes. The model, however, cannot be extended to multi-hop networks.

In \cite{7121036}, throughput and stability results of single-hop CMSA/CA network were obtained, but
it is restricted due to the assumption that nodes freeze their arrival process during a back-off period \cite{shneer2015stability}. Furthermore,
their analysis is only applicable for single-hop communication. In \cite{begin2016performance}, a two-level modeling approach was
used to model multi-hop wireless communication in which the first level is a slightly modified version of \cite{Bianchi} and the higher level model consists of
M/M/1/K queues representing each node in the network. The predication modeling framework also covers scenarios where two flows in opposite directions exist. 
Despite the claim that their model can handle any number of nodes, the scenarios they covered is limited to only four nodes. Additionally, routing is not considered in their study.

In \cite{vural2010multihop}, greedy distance maximization model, similar to our furthest neighbor routing policy, was proposed for 2-D randomly deployed sensor networks. They obtained a transformation of Gamma distribution to accurately model multi-hop distance distribution. However, their approach can only be used for furthest neighbor routing policy.
In \cite{ref22}, the author provided an analytical expression for maximum 
end-to-end throughput by building a conflict graph of a path. Although 
maximum end-to-end throughput of a chain with randomly distributed 
nodes is given in this paper, the expression requires finding clique number 
of the conflict graph which itself is an NP-Hard problem. The author also 
suggested using the maximum forward degree of the conflict graph instead of 
clique number, but no simulation was given to confirm the validity of the 
model. Moreover, the hidden and exposed node problems, the difference between transmission 
and interference range and, in general, 
 MAC related problems were neglected. The conflict graph modeling was also used in \cite{stojanova2017conflict} to 
provide a conceptually and computationally simpler model for unsaturated and single-hop wireless networks.
Despite the accuracy of their model, it cannot be extended to multi-hop wireless networks and it cannot certainly deal with routing.
As we mentioned earlier and to the best of our knowledge, no research has 
taken routing policy, node distribution and MAC layer limitation into account 
simultaneously. Moreover, no analytical results are available about maximum 
throughput in $2$-D networks.

\section{Analysis of Routing Policy in 1-D networks}
\label{routingAnalysis}

In this section, average end to end throughput is obtained considering a  
perfect MAC layer. Note that in this paper throughput is defined as the end-to-end throughput of the entire flow which is measured by finding the receiving rate of a destination node.

\begin{definition} 
Perfect MAC layer is defined as a MAC layer in which no 
hidden or exposed node problem occurs. Additionally, after an intermediate 
node sends its packet, it will wait until the packet reaches outside its 
interference range before sending another packet.
\end{definition} 

We will later extend our model 
to support imperfections in MAC layer in the next section. Let denote the 
throughput of single-hop connection by $C$. Now, we obtain the 
end-to-end throughput of multi-hop connection for two different routing policies
referred to as random neighbor routing and furthest neighbor routing. 

In random neighbor routing, next hop is selected randomly among 
the nodes in the transmission range of the current node toward a destination. 
Note that, the next hops are typically selected 
by a routing protocol but, in nodes’ point of view, next hops may have some characteristics 
which categorize that routing protocol as a random neighbor routing.

For example, in Fig. \ref{fig1}, node $i$ selects one 
of the nodes in a set $S=\{i+1,i+2,i+3,i+4\}$ with equal probability.
On the other hand, in furthest neighbor routing, intermediate nodes select 
the furthest node in its transmission range toward destination, i.e. node $i+4$ 
is selected as the next hop.
In 1-dimensional networks, the end-to-end throughput for this 
routing mechanism can be used as an upper-bound. 

Needless to say, in $2$-D networks when several flows are being carried out, it is 
possible for longer routes avoiding the congested region to outperform shorter 
routes even if the furthest neighbor routing is used. Particularly, when 
a routing mechanism in which the routes tend to be the shortest one is used, 
most of the network traffic pass through the center of the network which makes 
the central portion of network highly congested \cite{mei2009routing}. 
That is why even using the furthest neighbor routing mechanism which may pass 
through the congested central part of the network may be less efficient than a long 
route avoiding the center of the network.

\begin{figure}
\centering
\includegraphics[width=0.5\columnwidth]{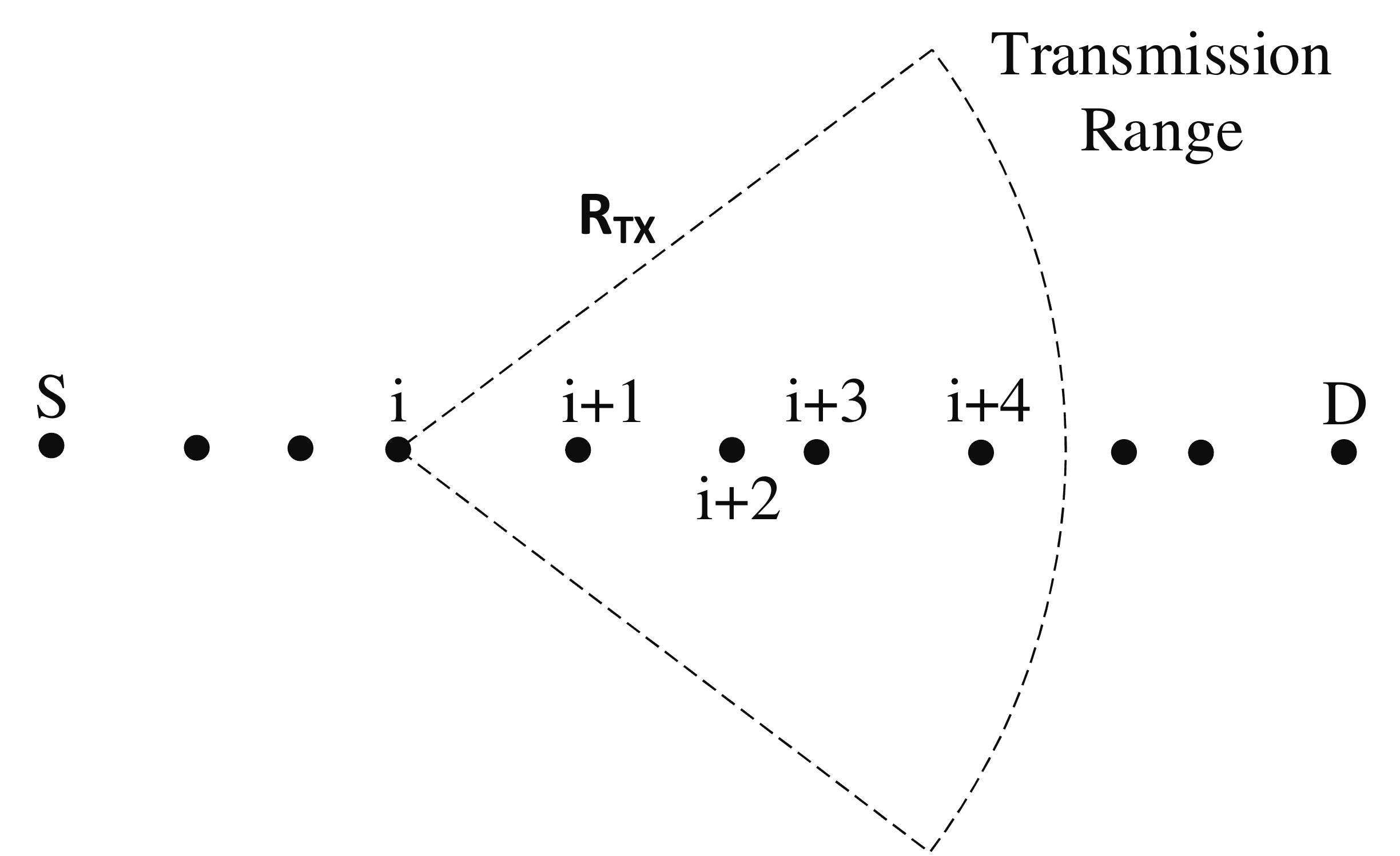}
\caption{In random neighbor routing, node $i$ selects one of the 
nodes $i+1$, $i+2$, $i+3$ or $i+4$. In furthest neighbor routing, node $i$ selects $m$.}
\label{fig1}
\end{figure}

Note that current well-known routing protocols such as AODV are not 
necessarily a random neighbor routing. Although jittering mechanisms~\cite{rezaei2016impact,rezaei2017general}, such as Uniform Jittering, 
may contribute to AODV's randomness, the fact that middle nodes can forward 
a RREQ packet more than once if the newly received RREQ packet belongs to 
a shorter path, renders this protocol less random. If intermediate nodes 
were not allowed to send RREQ packets more than once and they were just 
allowed to forward the first RREQ they had received, AODV would be 
construed as a random neighbor routing. 
It is also quite obvious that AODV is not equal to the furthest neighbor routing either. 
As long as hop count is concerned, AODV is better than random neighbor routing but worse than furthest neighbor routing. This fact will be clearly shown in Section \ref{simulation} by simulation.

The rest of this section categorized into three parts. In the 
first part, a general approach towards the analysis 
of end-to-end throughput is proposed. Given the information about routing policy and node distribution, the proposed analysis can be used to obtain end-to-end throughput of any routing policy. More precisely, our analysis requires the average number of hop needed to reach at a distance $x$ away from a source to be known.
Then, in the following subsections, the end-to-end throughput of random 
neighbor routing as well as furthest neighbor routing are obtained 
based on our general approach.

\subsection{General Approach}

Assume that node $i$ just forwarded its packet to the node $i+2$, 
depicted in  Fig. \ref{fig1}. Since they are 
located in the transmission range of each other, they cannot send their packets simultaneously. 
That is why the throughput of multi-hop transmission is less than 
single-hop transmission. 
Hence, we are interested in finding the average number 
of consecutive nodes following $i$ in the path that must forward the 
packet until the packet reaches a point where its transmission does not 
interfere with the transmission of $i$ so that it can send another packet.

For simplicity, we assume that all nodes have the same transmission range, 
$R_{tx}$, interference range, $R_i$, and carrier sensing range, $R_{cs}$, 
defined in \cite{ref1}. 
Note that in a perfect MAC layer, hidden or exposed node problems do not exist. Hence, in this section, the value of $R_{cs}$ 
is not taken into consideration, i.e. $R_{cs}$=$R_i$. Consequently, two 
transmissions can be simultaneously carried out when the sender and 
receiver of two connections are not in the interference range of each other.

In Fig. \ref{fig2}, assume that the node $i$ has just forwarded 
its packet. Let $N(x)$ be the average number of nodes that 
must forward the packet to reach at a distance $x$ away from 
a node. In order for $i$ to forward its next packet, the previous 
packet must travel outside of the interference range of node $i+1$. 
Otherwise, node $i+1$ fails to receive packets of $i$. Hence, the total 
number of transmissions needed so as to send the next packet is 
$1+N(R_i)$. Hence, the average value of maximum throughput is 
obtained from Eq. (\ref{eq::maxThrou}).

\begin{equation}
\label{eq::maxThrou}
T_{max} = \frac{C}{1+N(R_i)}
\end{equation}
where C is a throughput of a single hop communication. In perfect MAC layer, $T_{max}$ is the maximum 
throughput which can be used in upper layer for regulating source's 
sending rate. Needless to say, there would be no increase in end-to-end 
throughput if a source increases its sending rate. In reality, with 
imperfect MAC, increasing sending rate even deteriorates the end-to-end 
throughput due to high contention and collision. On the other side, if the 
source's sending rate is less than $T_{max}$, the end-to-end throughput 
is equal to source's sending rate. Therefore, the end-to-end throughput is 
obtained from Eq. (\ref{BaseFomula}) where $r$ is the source's sending 
rate.

\begin{figure}
\centering
\includegraphics[width=0.6\columnwidth]{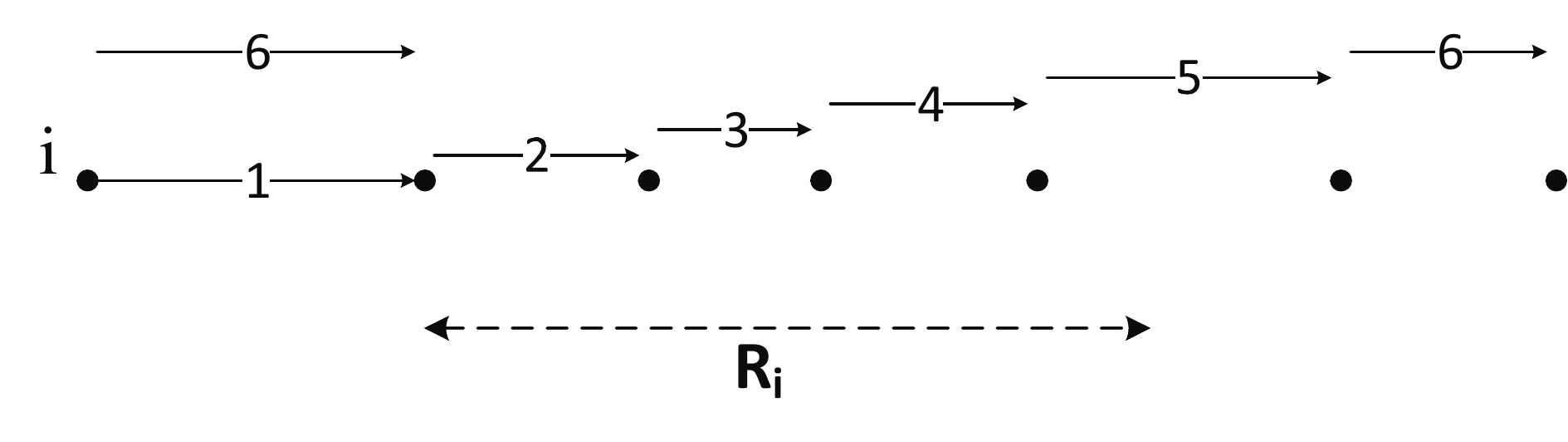}
\caption{Spatial reuse of wireless channel: After the fifth transmission,
 node $i$ can send a packet again.}
\label{fig2}
\end{figure}

\begin{table}[ht]
\caption[Table of notations]{Table of notations. }
\resizebox{1\textwidth}{!}{\begin{minipage}{\textwidth}
\begin{tabular}{ l | l  }
  $r$ & Sending rate of a source node \\
  $C$ & Throughput of a single-hop communication \\
  $T_{max}$ & Maximum throughput of multi-hop communication  \\
  $T(r)$ & Maximum throughput of multi-hop communication when source sending rate is $r$  \\
 $R_{tx} \ (R) \ $ & Transmission range\\
  $R_i$ & Interference range  \\
  $R_{cs}$ &  Carrier sensing range \\
 $N(x)$ & Average number of nodes to reach at a distance $x$ for general routing policy \\
 $\bar{N}(x)$ & Approximation of $N(x)$ \\
 $N_R(x)$ & Average number of nodes in random neighbor routing to reach at $x$ in $1$-D networks \\
 $N_F(x)$ & Average number of nodes in furthest neighbor routing to reach at $x$ in $1$-D networks \\
 $N_{R2D}(x)$ & Average number of nodes in random neighbor routing to reach at $x$ in $2$-D networks \\
 $N_{F2D}(x)$ & Average number of nodes in furthest neighbor routing to reach at $x$ in $2$-D networks \\
  $\lambda$ & Node density \\
  $d_{i}$ & Distance of node $i$ from the source\\
  $P_{col}$ & Probability of collision \\
 $x$ & Normalized air time \cite{ref1} \\
$\gamma_{n}$ & Normalized lower incomplete Gamma function \\
$\psi (x)$ & Function defined by Eq. (\ref{eq::psi}) for convenience \\
$C_n$ & Constant terms of $N_F(x)$ defined by Eq. (\ref{eq::Cn}) \\
\end{tabular}
\label{notations}
\end{minipage} }
\end{table}

\begin{equation}
T(r) = Min \Bigl\{ r,\frac{C}{1+N(R_i)} \Bigr\}
\label{BaseFomula}
\end{equation}

Now that we have the general formula to obtain the end-to-end throughput,
we need to compute $N(x)$. This function depends only on routing mechanism 
and node distribution. In this section, we assume that location of nodes in a network follows a 
homogeneous $1$-dimensional Poisson Point Process (PPP), that is, nodes are distributed randomly in a segment of a line as it is shown 
in Fig. \ref{fig2}. 

\subsection{Random Neighbor Routing}

In the random neighbor routing, a next hop is selected from neighboring 
nodes toward the destination with equal probability. In Fig. \ref{fig1}, 
the intermediate node $i$ selects one of the nodes in a set 
$S=\{i+1,i+2,i+3,i+4\}$ with probability of $\frac{1}{4}$. Here, we 
are interested in the distance distribution  to the next hop in 
a random neighbor routing.

\begin{theorem}
The distribution of distance to the next hop when using the random neighbor 
routing in a PPP network has a uniform distribution, $U{\sim}[0,R_{tx}]$.
\end{theorem}

\begin{proof}
In random neighbor routing, the next hop is selected randomly 
from the set of all nodes in the transmission range of an intermediate 
node, i.e. toward the destination, with an equal probability. Given that 
there are $n$ nodes in the transmission range of an intermediate node, 
we are interested in finding the distance distribution of each of them. 
It is known \cite{ppp} that in order to have a homogeneous PPP, in a 
region of length $R_{tx}$, location of points should be selected from a 
uniform random variable (RV), $U{\sim}[0,R_{tx}]$. Hence, the distance distribution of 
all $n$ nodes in the transmission range of an intermediate node are 
identically independently distributed (i.i.d.).
In other words, next hop is selected randomly with an equal probability 
from a pool of $n$ i.i.d uniform random variables. Therefore, regardless of 
the node selected, the distribution of the distance to that node is uniform 
RV, $U{\sim}[0,R_{tx}]$.

Note that a next hop in the random neighbor routing is not the same as the next immediate node. In fact, it is well known that the distance between 
two adjacent nodes in PPP follows Exponential distribution \cite{ppp}. However, the next hop in random neighbor routing can be any node located in the transmission range of node $i$ towards the destination.
\end{proof} 

The theorem is also true for binomial point process (BPP) since PPP is 
equivalent to BPP, conditioned on the presence of $n$ nodes in the 
transmission range \cite{srinivasa2010distance}. Consequently, all the 
formulas obtained in this paper for uniform random policy in PPP is also 
true in BPP.

As it is shown in Fig.~\ref{fig2}, in order for the node $i$ to send its next 
packet, the last packet should be forwarded until it reaches a node whose 
interference range is smaller than its distance to the node $i+1$. 
In other words, we are interested in finding how many uniform i.i.d random 
variables are required so that the sum of their values exceeds the interference range.

Given that the distance distribution to the next hop in random neighbor 
routing is uniform, $N(x)$ is given by theorem \ref{RV}. Hereafter, we will write 
$R$ instead of $R_{tx}$ for brevity.

\begin{theorem}
\label{RV}
The expected number of i.i.d uniform random variables ($U{\sim}[0,R]$) 
whose sum exceeds $x$ is given by Eq. (\ref{uniformSum}).

\begin{equation}
\label{uniformSum}
N_R(x)= \sum_{k=0}^{\lceil \frac{x}{R} \rceil - 1} \frac{(-1)^k}{k!} (\frac{x}{R}-k)^k e^{(\frac{x}{R}-k)}
\end{equation}

\end{theorem}

\begin{proof}
We are using induction to prove the theorem. For $0{<}x{\leq}R$, 
the solution is given in \cite{ross} which is $e^{\frac{x}{R}}$. 
Now, assuming that the formula holds for $0{<}x{\leq}nR$, we 
prove it also holds for $nR{<}x{\leq}(n+1)R$ in which $n$ is a 
positive integer. Let's define $M(x)$ as follows \cite{ross}

\begin{equation}
	M(x) \triangleq min  \Bigl\{ n:\sum_{i=1}^{n}U_i>x  \Bigr\}
\end{equation}

and 

\begin{equation}
	N_R(x)=E[M(x)].
\end{equation}

$N_R(x)$ can be obtained by conditioning on $U_{1}$, as the first 
uniform random variable, as follows

\begin{multline}
	N_R(x)= \int_{0}^{R} E[M(x)|U_1=y] f_{U_i} (y) dy =  \frac{1}{R} \int_{0}^{R} E[M(x)|U_1=y] dy.
\label{eqUni1}
\end{multline}

The conditional expectation itself is obtained from Eq. (\ref{eqUni2}).

\begin{figure*}[t!]
\begin{center}
 \subfloat[Node density = $0.04$]{ \includegraphics[width=.34\linewidth]{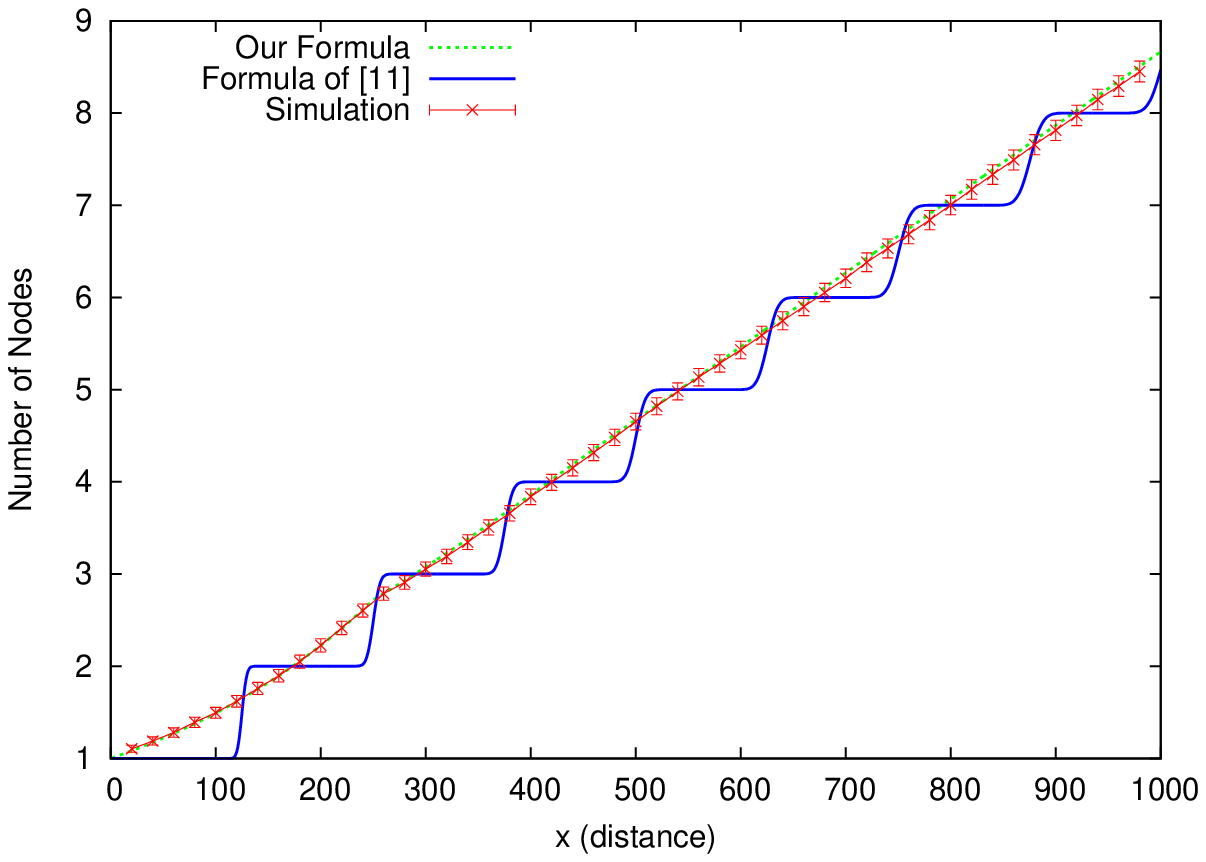}}
 \subfloat[Node density = $0.12$]{ \includegraphics[width=.34\linewidth]{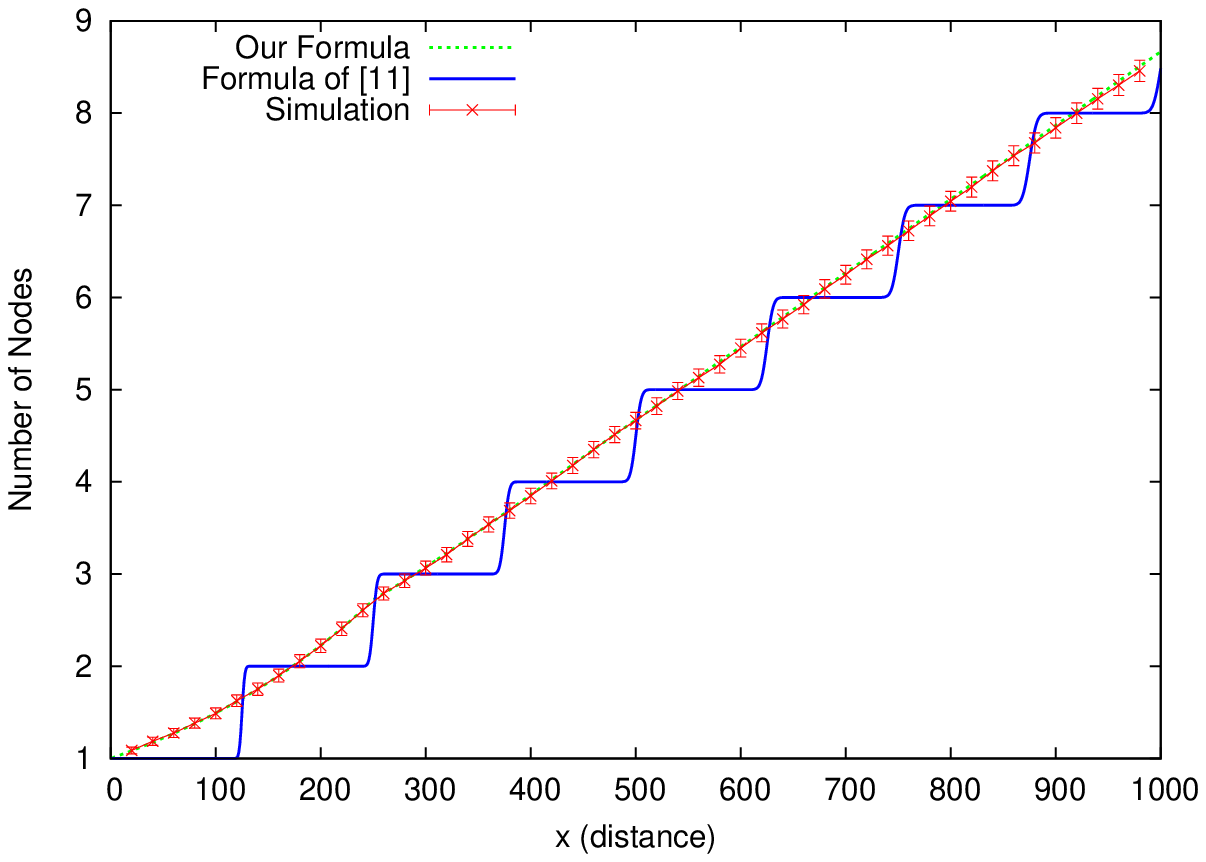}}
 \subfloat[Node density = $0.4$]{ \includegraphics[width=.34\linewidth]{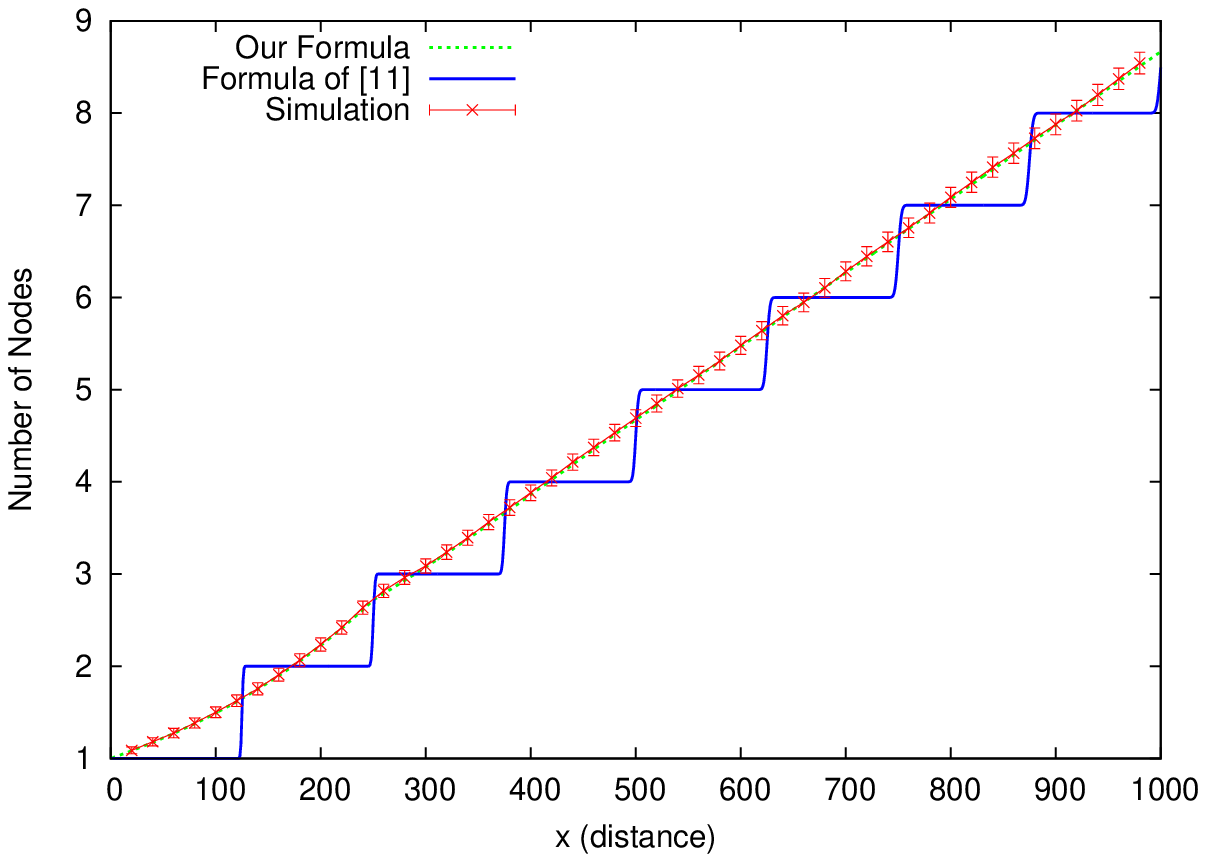}}
\caption{Average hop required to reach a distance (x) using random neighbor routing.}
\label{figRandomtDistFormula}
\end{center}
\end{figure*}

\begin{equation}
	 E[M(x)|U_1=y] = 
	\left\{
		\begin{array}{ll}
			1  & \mbox{if } x < y \leq R \\
			1 + N_R(x-y) & \mbox{if } y < x
		\end{array}
	\right.
\label{eqUni2}
\end{equation}

Substituting Eq. (\ref{eqUni2}) into Eq. (\ref{eqUni1}) and letting  $u=x-y$ results in Eq. (\ref{finalEq}).

\begin{equation}
\label{finalEq}
	N_R(x)= 1+ \frac{1}{R} \int_{x-R}^{x} N_R(u) du 
\end{equation}

By differentiating the last equation with respect to $x$, i.e. Eq. (\ref{finalEq}), 
we are given Eq. (\ref{uniformODE}).

\begin{equation}
\label{uniformODE}
	N_R^{\prime}(x)= \frac{N_R(x) - N_R(x-R)}{R}.
\end{equation}

Considering that $nR{<}x{\leq}(n+1)R$, it is easy to see that 
Eq. (\ref{uniformSum}) holds in Eq. (\ref{uniformODE}).

\end{proof} 

In \cite{rahmatollahi2012closed}, authors obtained a closed-form distribution 
of number of intermediate nodes required to reach a node at distance $D$ away 
from the source using incomplete gamma function as Eq. (\ref{RahmatollahiEq}). 

\begin{equation}
\label{RahmatollahiEq}
	N_H \sim \gamma_{n}(D;\lambda, n\beta) - \gamma_{n}(D;\lambda,(n+1)\beta)]
\end{equation}

In Eq. (\ref{RahmatollahiEq}), $\beta=1+\bar{d}\lambda$, $\lambda$ 
and $D$ are node density and distance, respectively. $\bar{d}$ is the average hop length 
under an arbitrary routing policy and $\gamma_{n}(x;\lambda, n)$ 
is normalized lower incomplete Gamma function defined as Eq. (\ref{eq::gamma}).

\begin{equation}
\label{eq::gamma}
	\gamma_{n}(x;\lambda, n) \triangleq \frac{1}{\Gamma(n)} \int_{0}^{\lambda x} t^{n-1}e^{-t} dt.
\end{equation}

 Note that this formula works for any routing policy, such as random neighbor routing and furthest neighbor routing. The only difference is the average hop length which clearly differs from one routing policy to another. Assuming an imaginary node is located at distance $D$ away from 
the source, we can use their formula to obtain $N_R(x)$ by finding 
the expected value of the distribution. Note that for random neighbor 
routing $\bar{d}=\frac{1}{2}R$. The formula of \cite{rahmatollahi2012closed}, however, is not as accurate as 
our formula, as it is shown in Fig. \ref{figRandomtDistFormula}.
For the simulation, we used a simple C++ code which randomly distributed nodes in a line of length $1250$ meter. Then, a source node was put at the beginning of the line. Finally, random neighbor routing policy was used to find how many nodes are required to reach certain points. We conducted $2000$ simulation for each data point.
Moreover, we omitted the simulations in which the network was disconnected which only happened at low densities. It is also clear that density does not have any effect on the number of nodes needed to reach a point as long as the network is connected.  Note that it is only true for random neighbor routing.

%

\subsection{Furthest Neighbor Routing}

The number of nodes in an area of length $R$ in PPP is Poisson random variable 
with parameter $\lambda R$ in which $\lambda$ is node density \cite{ppp}. 
Hence, the complementary cumulative distribution of distance, conditioned on 
having at least one node in an area of length $R$ could be obtained from
Eq. (\ref{DistanceCCD}). This equation indicates the probability of having 
at least one node in an area of length $R-x$, conditioned on having at least one 
node in the whole $R$. From Eq. (\ref{DistanceCCD}) we can conclude 
Eq. (\ref{eq::Conclude}) and as a result Eq. (\ref{equFurthestDistribution})
is derived, which is similar to the equation proposed in \cite{haenggi2005distances} that did not have the density term, $\lambda$.

\begin{equation}
\label{DistanceCCD}
P(X>x) = \frac{1-e^{-\lambda (R-x)}}{1-e^{-\lambda R}}
\end{equation}

\begin{equation}
\label{eq::Conclude}
	F_X(x) = 1 - P(X>x) = \frac{e^{\lambda x}-1}{e^{\lambda R}-1}
\end{equation}

\begin{equation}
\label{equFurthestDistribution}
	f_X(x) = \frac{\lambda e^{\lambda x}}{e^{\lambda R}-1}.
\end{equation}

Here, to make our model mathematically tractable, we assume 
that distribution of distance to next node for node $i$ is independent 
of the previous nodes. However, for furthest random routing, this is 
not always true, particularly when the network has low density. If the 
furthest node of $i$ is node $i+1$ in a distance $d_{i+1}$, we know 
that there is no node in distance $d_{i+1}<x<R$. Hence, the furthest 
node of node $i+1$ cannot reside in this region. That is why in furthest 
neighbor routing the distribution of distance to the next hop of intermediate 
nodes is not independent. It is worth noting that, in a relatively dense network, 
this region becomes so small that it barely has any perceptible influence 
on the analysis.

\begin{figure*}[t!]
\begin{center}
 \subfloat[Node density = $0.04$]{ \includegraphics[width=.34\linewidth]{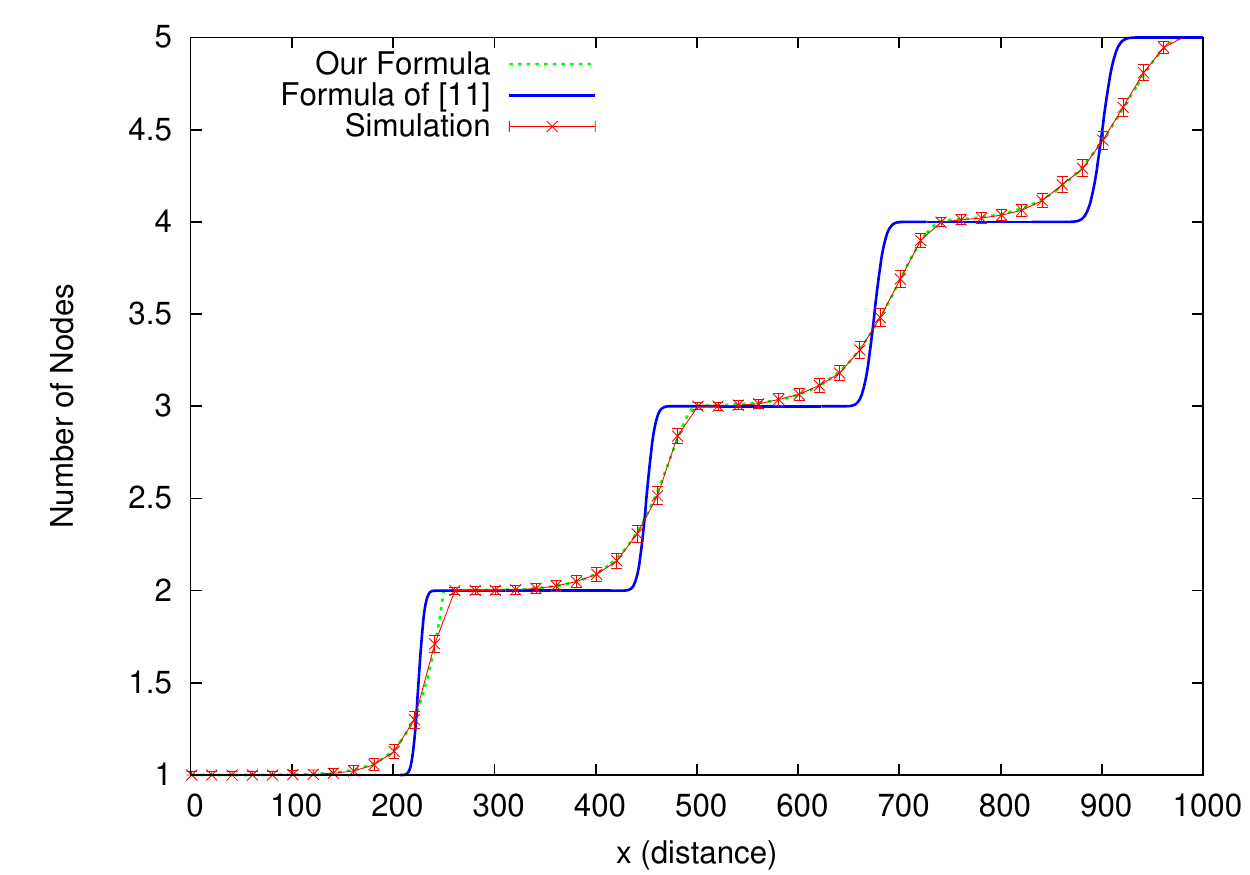}}
 \subfloat[Node density = $0.12$]{ \includegraphics[width=.34\linewidth]{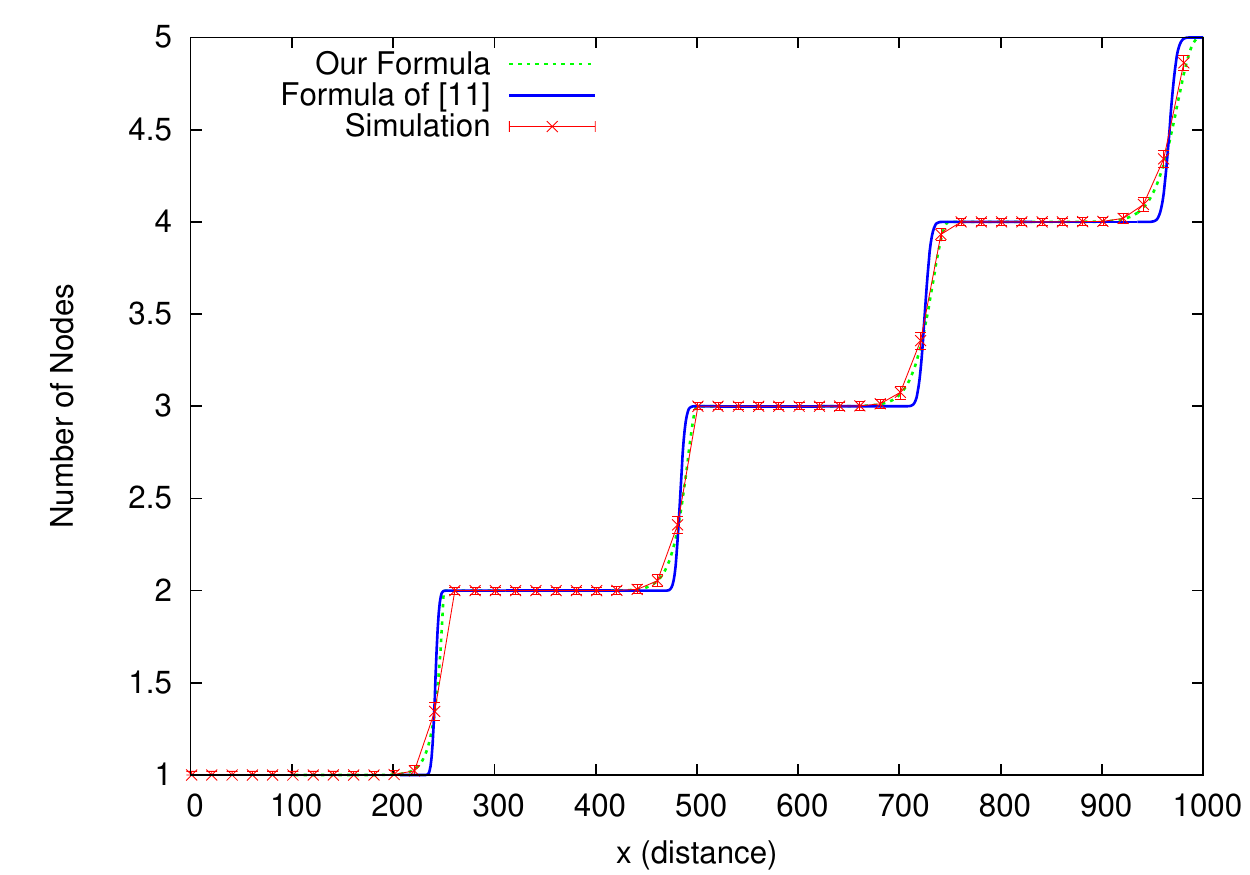}}
 \subfloat[Node density = $0.4$]{ \includegraphics[width=.34\linewidth]{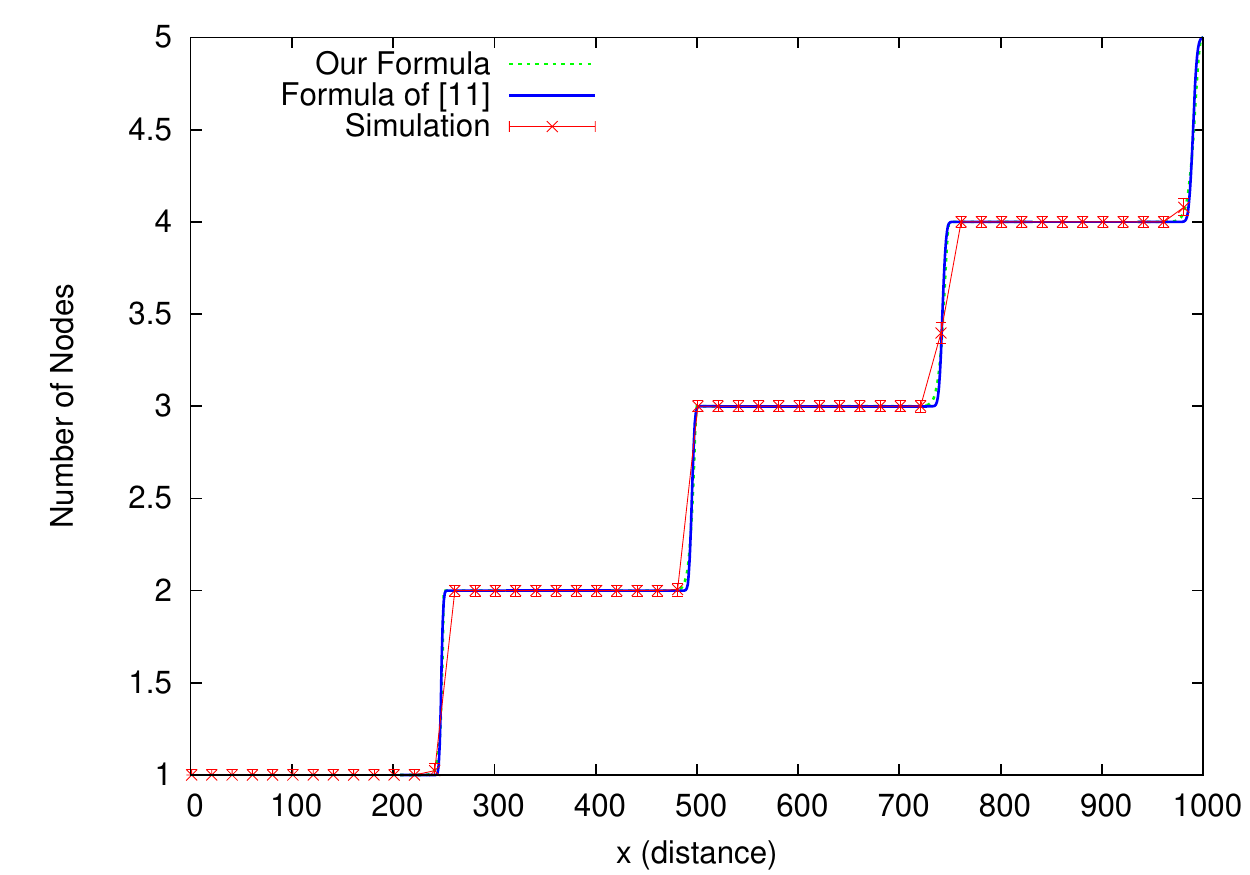}}
\caption{Average hop required to reach a distance (x) using furthest neighbor routing.}
\label{figFurthestDistFormula}
\end{center}
\end{figure*}

\begin{theorem}
\label{theorem::3}
The expected number of i.i.d random variables, with distribution function $f_X(x)$, whose sum exceeds 
$x$ is given by Eq. (\ref{furthestSum}).

\begin{multline}
\label{furthestSum}
	N_F(x)= C_{\lceil \frac{x}{R} \rceil} e^{\psi (x)} + {\lceil \frac{x}{R} \rceil}(1-e^{-\lambda R}) + 
	\sum_{k=1}^{\lceil \frac{x}{R} \rceil -1} \frac{(-1)^k \psi (x) \psi^{k-1} (x-kR) }{k!}  e^{\psi (x-RK)} C_{\lceil \frac{x}{R} \rceil - k},
\end{multline}

In Eq. (\ref{furthestSum}), parameters $\psi (x)$ and $C_n$ are defined 
by Eq. (\ref{eq::psi}) and (\ref{eq::Cn}), respectively.

\begin{equation}
\label{eq::psi}
	\psi (x) \triangleq \frac{\lambda x}{1-e^{-\lambda R}}
\end{equation}

\begin{equation}
\label{eq::Cn}
	C_n \triangleq 
	\left\{
		\begin{array}{ll}
			e^{-\lambda R}  & \mbox{if } n=1 \\
			e^{\psi (-R)}[\psi(R) C_1+e^{-\lambda R}-1] + C_1 & \mbox{if } n=2 \\
			e^{\psi (-R(n-1))}(e^{-\lambda R}-1) + C_{n-1} + \\ \sum\limits_{k=1}^{n-2} \bigg[ \frac{\psi (R(n-1)) \psi^{k-1} (R(n-1-k)) }{(-1)^k k! e^{\psi (Rk)}}  \times  (C_{n-k-1}-C_{n-k})  \bigg]   & \mbox{if } n>2
		\end{array}
	\right.
\end{equation}

\end{theorem}

\begin{proof}
Just like what was described in Theorem \ref{RV}, we use induction 
to prove the theorem. Let $M(x)$ and $N_F(x)$ be as Eq. (\ref{eq::Mx})
and Eq. (\ref{eq::NFx}), respectively. 

\begin{equation}
\label{eq::Mx}
M(x) \triangleq min  \Bigl\{ n:\sum_{i=1}^{n}f_i>x  \Bigr\}
\end{equation}

\begin{equation}
\label{eq::NFx}
N_F(x)=E[M(x)].
\end{equation}

It is notable that $N_F(x)$ can be obtained by conditioning on $U_{1}$ 
as it is shown in Eq. (\ref{eq::NFxCond}), while the conditional 
expectation itself is obtained form Eq. (\ref{eq::condEx}).

\begin{multline}
\label{eq::NFxCond}
	N_F(x)= \int_{0}^{R} E[M(x)|f_1=y] f_{i} (y) dy =  
		\frac{\lambda}{e^{\lambda R}-1} \int_{0}^{R} e^{\lambda y} E[M(x)|f_1=y] dy
\end{multline}

\begin{equation}
\label{eq::condEx}
	E[M(x)|f_1=y] = 
	\left\{
		\begin{array}{ll}
			1  & \mbox{if } x < y \leq R \\
			1 + N_F(x-y) & \mbox{if } y < x
		\end{array}
	\right.
\end{equation}

Letting $u=x-y$ results in Eq. (\ref{eq::NFxNew}) while 
Eq. (\ref{eq::expNFx}) could be obtained from multiplying 
both side by $e^{-\lambda x}$.

\begin{equation}
\label{eq::NFxNew}
	N_F(x)=  1+ \frac{\lambda e^{\lambda x}}{e^{\lambda R}-1} \int_{x-R}^{x} e^{-\lambda u} N_F(u) du.
\end{equation}

\begin{equation}
\label{eq::expNFx}
	 e^{-\lambda x} N_F(x)=  e^{-\lambda x}+ \frac{\lambda}{e^{\lambda R}-1} \int_{x-R}^{x} e^{-\lambda u} N_F(u) du.
\end{equation}

By differentiating Eq. (\ref{eq::expNFx}) with respect to $x$, we are given 
Eq. (\ref{FurthestODE}).

\begin{equation}
\label{FurthestODE}
	N_F^{\prime}(x)= \frac{\lambda N_F(x)}{e^{\lambda R}-1} + \lambda N_F(x) - \frac{\lambda e^{\lambda R} N_F(x-R)}{e^{\lambda R}-1} - \lambda.
\end{equation}

As a result of the definition of $N(x)$, for $0{<}x{<}R$, $N(x-R)$ 
is equal to zero. Note that if the value of $x$ is negative, 
there is no transmission needed to reach the distance x since the packet 
has already passed that point. Hence, we conclude Eq. 
(\ref{eq::NFxSemiFinal}) for $0{<}x{\leq}R$. The solution of this equation is 
mentioned in Eq. (\ref{eq::NFxFinal}). Since $N_F(0)=1$, $c_1$ 
would be equal to $e^{-\lambda R} $.

\begin{equation}
\label{eq::NFxSemiFinal}
	N_F^{\prime}(x)= \frac{\lambda N_F(x)}{e^{\lambda R}-1} + \lambda N_F(x) - \lambda =  \frac{\lambda e^{\lambda R} N_F(x)}{e^{\lambda R}-1} - \lambda 
\end{equation}

\begin{equation}
\label{eq::NFxFinal}
	N_F(x)= c_1 e^{\psi(x)} +1 - e^{-\lambda R} 
\end{equation}

Now, considering that the theorem holds for $(n-1)R{<}x{\leq}nR$, 
it is straightforward, yet tedious, to see that Eq. (\ref{furthestSum}) holds in Eq. 
(\ref{FurthestODE}) for  $nR{<}x{\leq}(n+1)R$. Thus, we omit the remainder of the cumbersome calculation and substitution here. Note that differential equation 
constant, $c_n$, is obtained by solving $N_F(R(n-1)+\varepsilon) = N_F(R(n-1))$.
\end{proof} 

Eq. (\ref{RahmatollahiEq}) can also be used to obtain average number 
of hops in furthest neighbor routing as it is suggested in 
\cite{rahmatollahi2012closed}. The value of $\bar{d}$, however, is different 
from the case of random neighbor routing. The authors suggested the following 
formula to obtain $\bar{d}$ of furthest neighbor routing.

\begin{equation}
	\bar{d} = \frac{1}{\lambda} ln  \left( 1-\frac{\lambda \bar{d}}{\lambda - \lambda \bar{d} -1}  \right).
\end{equation}

However, the implicit form of the equation which requires numerical 
evaluation increases the complexity of their solution even further 
\cite{vural2007probability}. Having compared our formula with the one proposed in 
\cite{rahmatollahi2012closed}, we have found that our 
approach is extremely accurate as it is shown in Fig. 
\ref{figFurthestDistFormula}. Regardless of the node density, 
our approach always agrees with simulation. 

Using Eq. (\ref{BaseFomula}) and Theorem \ref{theorem::3}, 
a source node can obtain the end-to-end throughput for furthest 
neighbor routing. Unlike previous works on end-to-end 
throughput \cite{ref4,ref5,ref6,ref8,ref9,ref10,ref14,ref17,ref13} which assumed 
that nodes are placed equi-distance apart or assumed the location 
of nodes to be known by the source, our analysis obtains end-to-end 
throughput based on realistic node distribution and routing policy.

\section{MAC Layer Consideration}
\label{macLayer}

In this section, limitations of wireless MAC layer are taken into account 
to obtain a more accurate expression for maximum end-to-end throughput. It is obvious that the perfect MAC, as defined in previous section, is impossible in wireless networks. Therefore, it is crucial to consider MAC and phyiscal layers' limitation alongside the routing and node distribution.
Most previous research has focused on MAC layer limitations without considering 
node distribution and routing policy, which clearly shows the importance of analyzing MAC layer.
Although analyzing MAC and routing layers simultaneously is difficult,
no analysis is precise enough if it overlooks one of the aspects. Thus, we provide a method 
to constitute a comprehensive model containing both routing and MAC layer consideration in 
this section. In this section, the probability of collision and hidden node problem, which is non-existence in perfect MAC layer, are taken into account to obtain a more accurate end-to-end throughput.

\begin{figure}
\centering
\includegraphics[width=0.4\columnwidth]{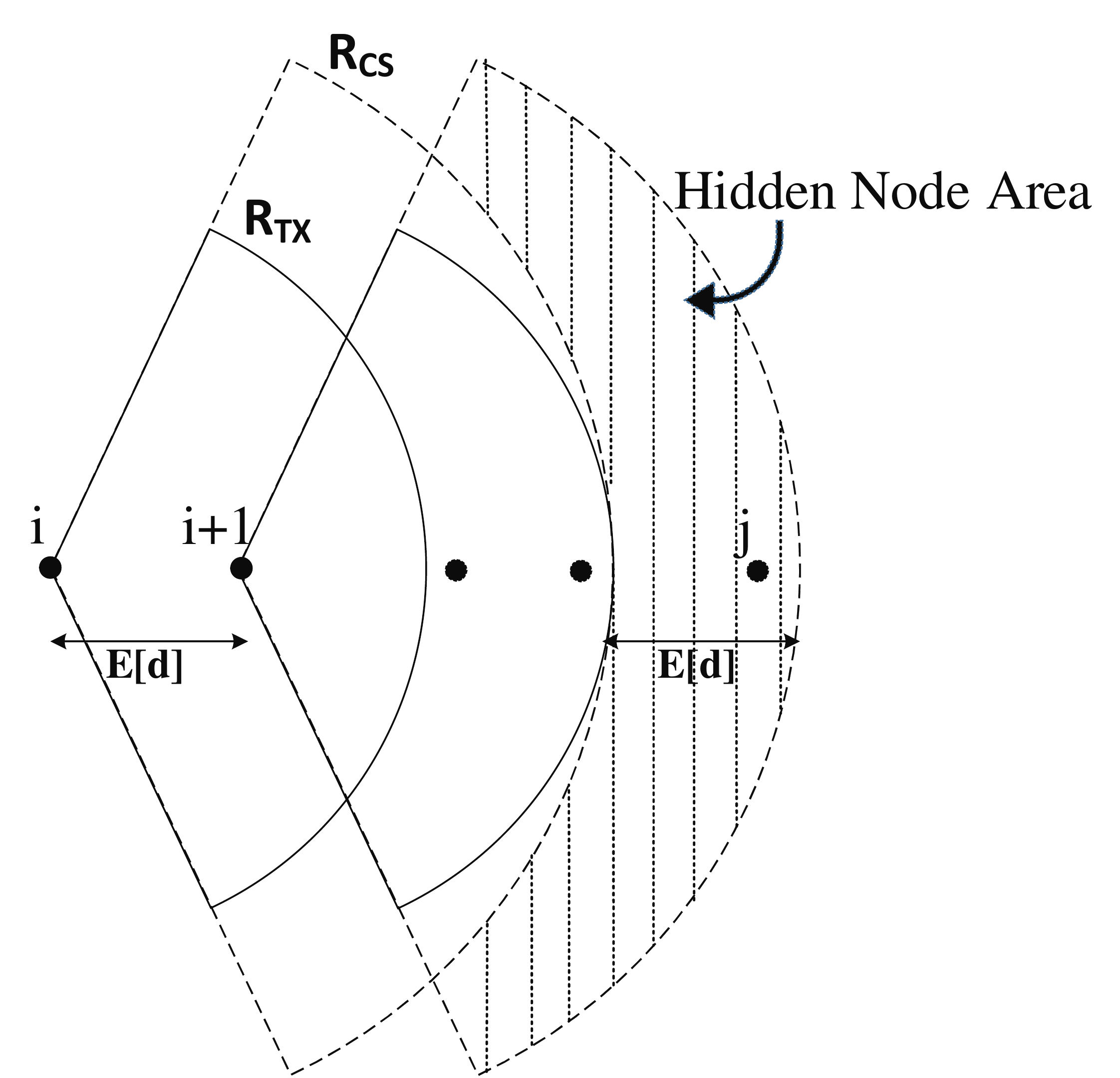}
\caption{Hidden node area.}
\label{figHNA}
\end{figure}

In wireless networks, no perfect MAC layer protocol has been proposed so far.
It is also unlikely to be able to devise one due to the inherent problems 
and difficulties of wireless media. Hence, our analysis in the previous section 
can be used as a maximum upper bound throughput. If some information about the MAC layer is available, it is possible to improve 
our analysis. As an instance, if the probability of collision is known, denoted by 
$P_{col}$, the number of transmission needed to successfully transmit a packet 
would be $\frac{1}{1-P_{col}}$. Therefore, maximum throughput can be obtained 
from Eq. (\ref{eq::maxThrough}), which is the improved version of Eq. (\ref{eq::maxThrou}). It is worth noting that, in this equation, the
assumption $R_{cs}>R_i$ is considered.

\begin{equation}
\label{eq::maxThrough}
T_{max} = \frac{1-P_{col}}{1+N(R_{cs}) }C
\end{equation}

Note that, in general, $R_{tx}<R_{i}<2R_{tx}<R_{cs}$  is 
satisfied \cite{ref13}. In this case, unlike perfect MAC layer, 
carrier sensing may block non-interfering transmission, referred to as exposed node problem \cite{ref17}. 
Consequently, it limits the capacity further. Hence, in order to 
make sure that an intermediate node can send its next packet, 
the previous packet must be forwarded to a distance greater than 
$1+N(R_{cs})$. That is why $N(R_{cs})$ is used in this formula 
instead of $N(R_i)$.

It is also notable that the value of $P_{col}$ is not independent of 
routing mechanism, nodes distribution and density. That is the main 
reason the other proposed analysis has always assumed a chain 
topology with constant distance or very particular structure. To obtain 
the value of $P_{col}$, both routing mechanism and MAC layer protocol
are taken into account here. We will use an IEEE $802.11$ MAC analysis 
based on \cite{ref1} and two different routing mechanisms elaborated 
upon in the previous section to obtain end-to-end throughput.

In IEEE $802.11$ MAC, nodes that are not located in the carrier sensing 
range of a sender, but located in the carrier sensing range of the receiver 
may cause a collision in that transmission \cite{ref1}. This region containing these 
interfering nodes, which is shown in Fig. \ref{figHNA}, is called hidden 
node area. Since the source is not aware of the transmission in this 
region, it may start a transmission. However, the receiver will fail to 
decode the packet due to the interference.

Since collision occurs due to nodes in hidden node area, we are 
interested in obtaining the number of nodes in this region. Given any 
arbitrary routing policy that selects next hop statistically independent 
from the previous nodes, we expect to see an equal number of nodes, 
chosen for routing, in each line segment of length $x$ on average. As 
it is shown in Fig. \ref{figHNA}, the length of hidden node area is equal 
to the distance between node $i$ and $i+1$, denoted by $E[d]$.
On average, we expect to see only one node in a region of 
length $E[d]$ since $E[d]$ is the average distance between 
two consecutive nodes. Thus, the average number of nodes 
in hidden node area is 1 regardless of the routing policy. 
Note that we are only interested in nodes selected for 
routing packets. There might be many 
non-contributing nodes in this area that does not have 
any influence on throughput.

%

\begin{figure*}[t!]
\begin{center}
 \subfloat[Random neighbor routing]{ \includegraphics[width=.45\linewidth]{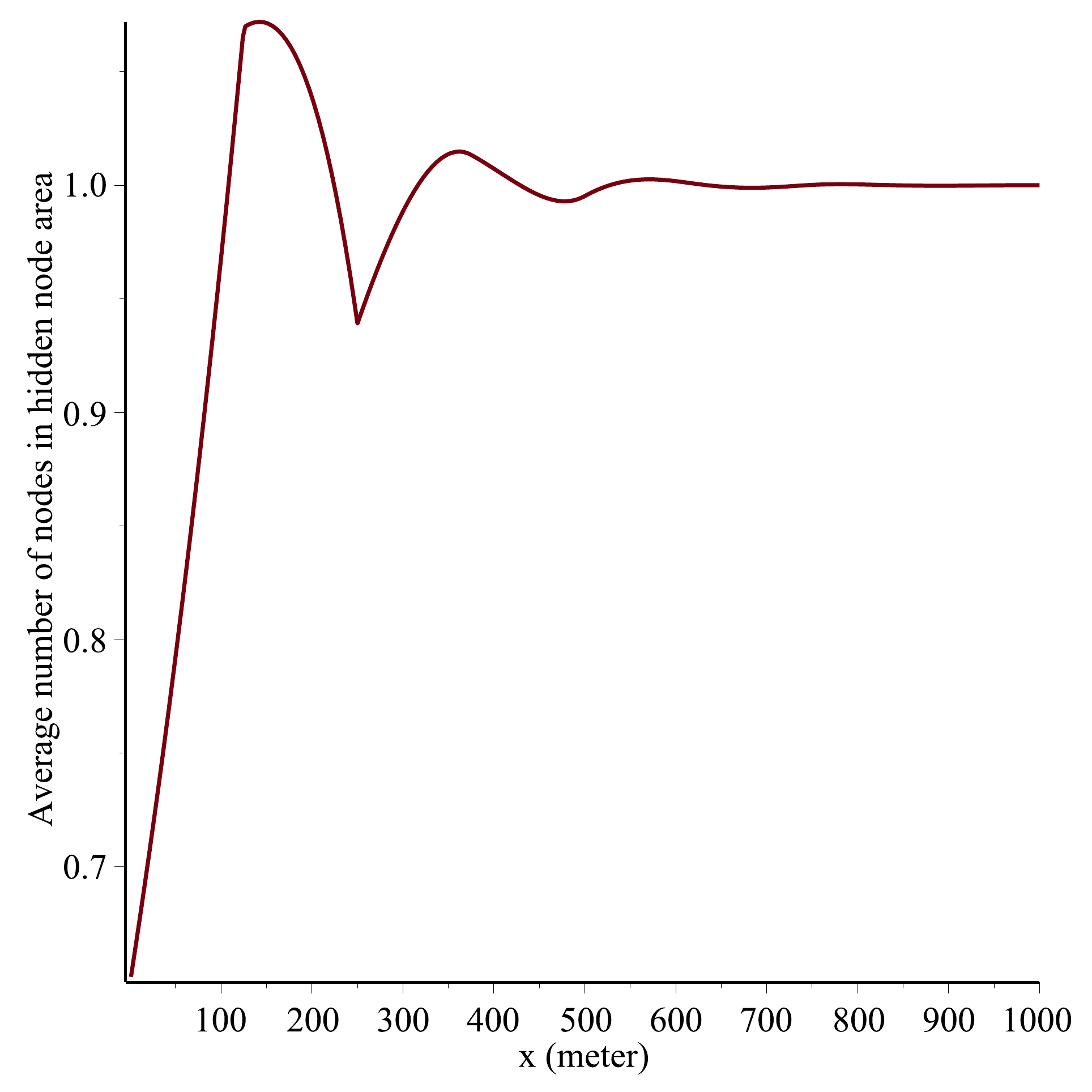}\label{HiddenAreaRandom}}
 \subfloat[Furthest neighbor routing]{ \includegraphics[width=.45\linewidth]{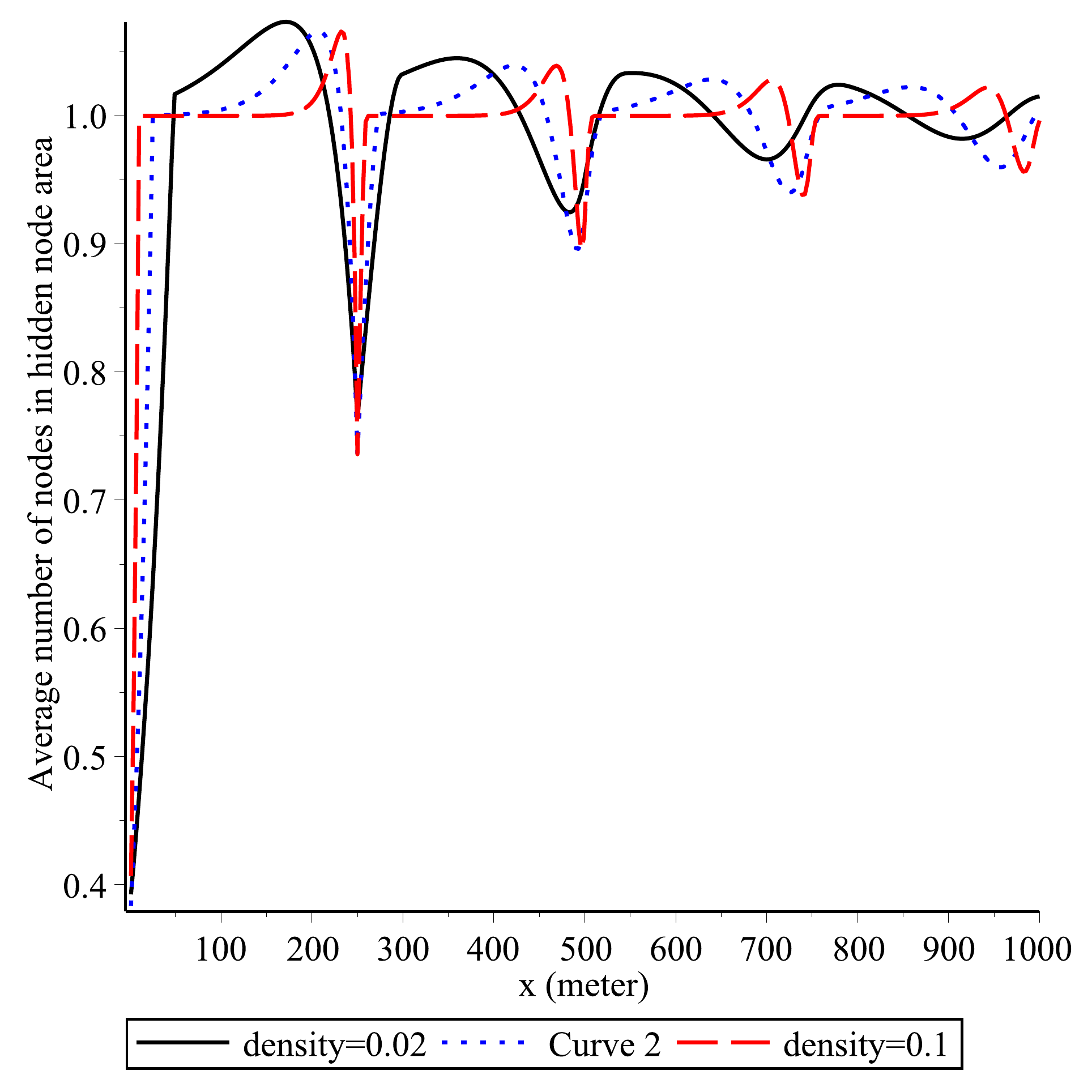}\label{HiddenAreaFurthest}}
\caption{Average number of nodes in hidden node area of length (E[d])}
\label{figHidArea}
\end{center}
\end{figure*}

Since this argument is not necessarily true for all statistical 
phenomena, we have computed number of forwarding nodes 
in each region of length $E[d]$. To obtain number of nodes in 
a region of length E[d], one can simply evaluate $N(x+E[d])-N(x)$. 
As it is shown in Fig. \ref{HiddenAreaRandom}, the number of nodes 
in hidden node area for random neighbor routing converges to $1$ and 
there is no significant variance over different $x$. This is also true for 
furthest neighbor routing illustrated in Fig. \ref{HiddenAreaFurthest}. 
Like random neighbor routing, the number of nodes in hidden node area 
converges to $1$ regardless of the node density. Although for the higher 
value of density, it is closer to 1, even for lower densities the variance 
is negligible. Nevertheless, starting from the source node, the hidden nodes must reside between the source's carrier sensing range ($R_{cs}$) and the next hop's carrier sensing range. In other word, even at the worse case scenario where the next hop is very close to the source node, hidden nodes are located at distances which are far greater than 250m  since  $R_{cs}>2R_{tx}>R_{tx}=250m$. The most variance occurs for small value of $x$ in furthest neighbor routing. Note that as x increases, $N(x)$ gets smoother, as it is depicted in Fig. \ref{figRandomtDistFormula} and \ref{figFurthestDistFormula}. Hence, $N(x+E[d])-N(x)$ also becomes smoother as x increases.
Hence, we obtain the probability of collision based on the 
fact that there is always one node in hidden node area on average.

Let $j$ be the node that is located in the hidden node area and $x_i$ be 
the normalized $airtime$ of node $i$ according to \cite{ref1}. 
Assuming that all nodes transmit at the same rate, we omit the index of $x_i$ and, 
hereafter, we write $x$ instead. If any node that located in carrier sensing 
range of both $i$ and $j$ is transmitting a packet, $i$ and $j$ cannot start 
sending a packet and consequently they cannot collide in this case. Such
intermediate nodes are called \textit{contending nodes}. The number of contending 
nodes of the node $i$ located at distance $d_i$ from the source would be 
$N(d_i+R_{cs})-N(d_i)$.

In order for a collision to happen, all contending nodes must be silent. Hence, 
this value should be removed from the sample space. The probability of collision 
for the node $i$, according to \cite{ref1}, is obtained from
Eq. (\ref{eq::collision}) in which $a$ is a fraction of time devoted to the data packet. 

\begin{equation}
\label{eq::collision}
P_{col}=\frac{ax}{1-[N(d_i+R_{cs})-N(d_i)]x}
\end{equation}

Here, the normalized $airtime$ represented by $x$ is equal to $\frac{r}{C}$. Note that in this formula $P_{col}$ depends on the node's distance from the source, $d_i$. However, we assumed that all nodes are statistically independent as a result of which their position does not have a considerable influence. Additionally, we will show in the next section that $N(d_i+R_{cs})-N(d_i)$ can be considered to be constant since the derivative of $N$ converges to a constant value. Hence, we will compute $P_{col}$ for the first node, the $d_i$ of which is zero.

\section{Approximation of $N(x)$}
\label{approximation}
In this section, we aim to provide an approximation for $N_R(x)$ and $N_F(x)$, and use 
it to obtain $T(r)$. As it is shown in Fig. \ref{figRandomtDistFormula}, it is easy to observe 
that for large enough $x$, $N(x)$ approaches a straight line. Additionally, 
it is also easy to see that, using root test, $N^{\prime}(x)$ given by 
Eq. (\ref{uniformODE}) and (\ref{eq::NFxSemiFinal}) is convergent. Hence, it is reasonable to find a linear approximation of $N(x)$.

\begin{theorem}
If $Y_1$, $Y_2$, \ldots are mutually independent, identically random variable with any arbitrary distribution (routing policy) that represents the distribution of the distance to the next hop, linear approximation of their corresponding $N(x)$ is as follows
\begin{equation}
\label{eq::GenerallinearApprox}
\bar{N} (x) = \frac{x}{E[Y]} + \frac{E[Y^2]}{2E^2[Y]} 
\end{equation} 
\end{theorem}

\begin{proof}
Using the same approach as it is used 
in Theorem (\ref{RV}), we get the following equation, independent of Y's distribution, when $0<x<R$,
\begin{equation}
\label{eq::GeneralProofStartingEq}
	N(x) = 1+ \int_{0}^{x} f_{Y}(x-u) N(u) du. 
\end{equation}

Using the Laplace transform and factoring $\mathcal{L}\{N(y)\}$, we get

\begin{equation}
\label{eq:MainLaplace}
	\mathcal{L}\{N\} (s) = \frac{1}{s(1-\mathcal{L}\{f_{Y}\} (s))}.
\end{equation}

To solve this differential equation, we need to know the Laplace transform 
of $f_{Y}$. Assuming that all moments of $f_{Y}$ exist, Laplace 
transform can be obtained by moment generating property as follows

\begin{equation}
	\mathcal{L}\{f_{Y}\} (s) = \sum_{n=0}^{\infty} (-1)^n \frac{s^n}{n!} E[Y^n].
\end{equation}

Using this Laplace transform in Eq. (\ref{eq:MainLaplace}), we are given

\begin{equation}
	\mathcal{L}\{N\} (s) = \frac{1}{s^2} \times \frac{1}{\sum_{n=1}^{\infty} (-1)^{n-1} \frac{s^{n-1}}{n!} E[Y^n]}.
\end{equation}

Now, since the summation is a power series, we can substitute 
the reciprocal of this power series as follows \cite{realAnalysis}

\begin{equation}
	\mathcal{L}\{N\} (s) = \frac{1}{s^2} \times ( \frac{1}{E[Y]} + \frac{E[Y^2]}{2E^2[Y]}s + ... ).
\end{equation}

Note that since we are interested in a linear approximation of 
$N(x)$, the Laplace form should be in a form 
$\mathcal{L}\{N\} (s)=as^{-2}+bs^{-1}$. As a result, 
we only consider the first two terms of the series, and 
the approximated Laplace transform is given by

\begin{equation}
	\mathcal{L}\{N\} (s) \approx  \frac{s^{-2}}{E[Y]} + \frac{E[Y^2]}{2E^2[Y]}s^{-1}. 
\end{equation}

It is easy to see that this equation is the Laplace transform of Eq. (\ref{eq::GenerallinearApprox}). 

\end{proof}

Given the fact that $N^{\prime}(x)$ approaches a constant value, 
the value of $[N(d_i+R_{cs})-N(d_i)]$ becomes independent of 
$d_i$ and $P_{col}$ becomes equal for all intermediate nodes. 
Hence, Eq. (\ref{eq::collision}) can be solved with $d_i=0$. Now that a simple approximation of $N(x)$ is provided, maximum throughput can be obtained by Eq. (\ref{eq::maxThrough}).

\subsection{Approximation of $N_R(x)$}

For random neighbor routing, linear approximation is obtained as follows

\begin{equation}
\label{uniformSumapproximation}
\bar{N}_R(x)= \frac{2x}{R} + \frac{2}{3}.
\end{equation}

Curve fitting using least squares error method
for $0<x<10R$ also gave almost 
the same value. The slope of the line obtained by least squares method 
was also $0.00798$ for $R=250$ which is almost the same as 
$N_R^{\prime}(x)$. Note that x-intercept, unlike the slope of the 
approximation line, is independent of the $R$.

\subsection{Approximation of $N_F(x)$}

Using the same method as random neighbor routing, the slope of the 
approximation function of $N_F(x)$ is obtained by its expected value. The expected value of the distribution function $f_x$ in Eq. 
(\ref{equFurthestDistribution}), denoted by $E[X_F]$, is shown in 
Eq. (\ref{eq::ExpDF}). 

\begin{equation}
\label{eq::ExpDF}
E[X_F]= \frac{e^{-\lambda R}+\lambda R - 1}{\lambda(1-e^{-\lambda R})}.
\end{equation}
The second moment of $X_F$ is as follows
\begin{equation}
\label{eq::ExpE2}
E[X^2_F]= \frac{1}{e^{\lambda R}-1} [ R^2 e^{\lambda R} - \frac{2R}{\lambda} e^{\lambda R} + \frac{2e^{\lambda R}}{\lambda^2}-\frac{2}{\lambda^2} ].
\end{equation}

Using these average and second moment of $X_F$, we can use Eq. (\ref{eq::GenerallinearApprox}) as a linear approximation.
 Usefulness of our linear approximations will be shown in Section \ref{simulation}.

\section{Maximum Throughput in 2-D Networks}
\label{$2$-DSection}

As it has been shown in previous sections, the exact formula 
to obtain $N(x)$ can get extremely complicated for some routing 
policies. Consequently, the approximation formulas introduced in 
Section \ref{approximation} are shown to be fairly accurate. 
In $2$-D networks, obtaining the exact value of $N(x)$ even gets 
more complicated, if possible at all. In fact, the probability distribution 
of the number of hops to reach the destination in $2$-D networks 
is derived in \cite{antunes2015hop} which not only does not have 
a closed-form structure, but also may have infinite unsolvable integrals. 
Hence, in this section, we obtain the approximate value of $N(x)$ 
which will be mathematically traceable.

\begin{figure}
\centering
\includegraphics[width=0.5\columnwidth]{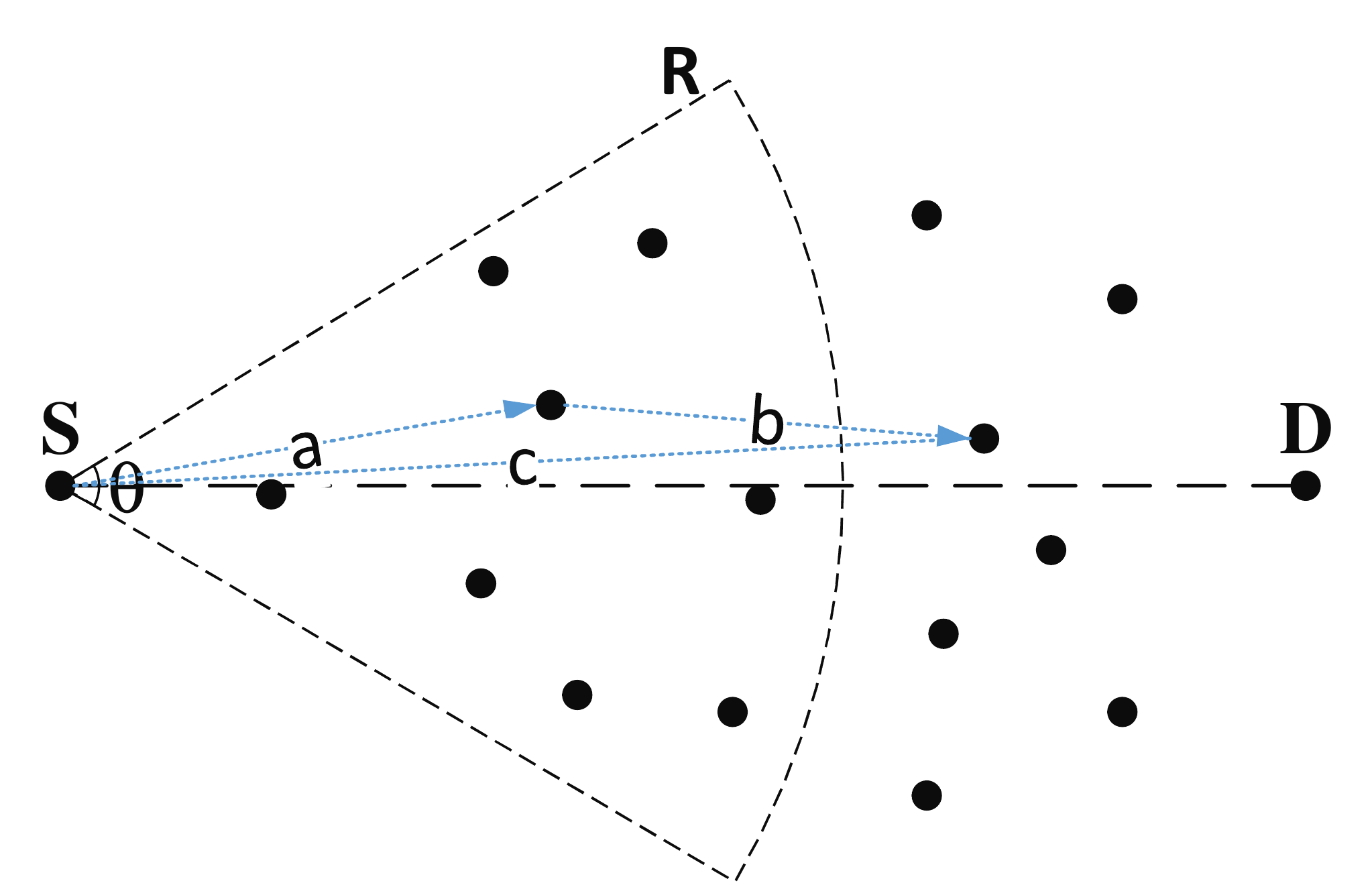}
\caption{Angle of progression and distance to the next hop in $2$-D 
networks. In $2$-D networks, sum of distances of two consecutive 
nodes are not necessarily equal to the distance from the first node 
to the third one.}
\label{fig2D}
\end{figure}

As depicted in Fig. \ref{fig2D}, The most problematic issue in $2$-D networks is that the sum of distances 
from the first to the second node, i.e. $||a||$, and the second to the third node, i.e. $||b||$, is not equal to the 
distance from the first node to the third one, i.e. $||c||$.
In fact, the angle of progression (AoP), which is defined as an angle of a sector in 
which the next hop is selected, $\theta$, should be taken into account. However, 
this consideration makes these random variables dependent and consequently makes 
the problem of finding $N(x)$ in $2$-D networks intractable. Since it is shown that 
routing protocols usually choose a straight line from a source to a destination if the 
network is dense enough \cite{straightLine}, we assume that these random variables 
are independent. Note that if the angle of progression is relatively small, these random variables will act independently. This assumption also allows us to use Eq. (\ref{eq::maxThrough}) to obtain maximum throughput.

This section includes two parts which aim to approximate $N(x)$ for random neighbor routing, $N_{R2D}(x)$, and furthest neighbor routing, $N_{F2D}(x)$.

\subsection{Approximation of $N_{R2D}(x)$}

First, we need to find the probability density of distance to 
the next hop when random neighbor policy is used, like our 
$1$-D analysis. The probability density function is denoted by 
$f_{R2D}$. Cumulative distribution function can be obtained 
by dividing the area of sector of length $x$ by the area of sector 
of length $R$, as follows

\begin{equation}
	F_{R2D}(x) = \frac{\frac{\theta}{2}x^2} {\frac{\theta}{2}R^2} = \frac{x^2}{R^2}.
\end{equation}

Therefore, density function and expected value are obtained as follows,

\begin{equation}
	f_{R2D}(x) = \frac{dF(x)}{dx} = \frac{2x}{R^2}
\end{equation}
and
\begin{equation}
	E[X_{R2D}] = \frac{2}{3}R.
\end{equation}

The second moment is also easy to get which is

\begin{equation}
	E[X^2_{R2D}] = \frac{R^2}{2}.
\end{equation}

Hence, the approximate number of nodes required to reach the location 
at $x$, $\bar{N}_{R2D}(x)$, is obtained as follows

\begin{equation}
\label{ApproximatedR$2$-D}
	\bar{N}_{R2D}(x) = \frac{3x}{2} + \frac{9}{16}.
\end{equation}

Now, using Eq. (\ref{eq::maxThrough}) and (\ref{ApproximatedR$2$-D}), 
we can obtain the end-to-end throughput of random neighbor routing in 
$2$-D networks. In the next section, we will show that this approximation 
works perfectly particularly when the angle of progression is small.

\subsection{Approximation of $N_{F2D}(x)$}

In this section, we will find the approximation for $N_{F2D}(x)$. The probability density of the distance to 
the furthest node in $2$-D networks obtained similar to the $1$-D 
case and it is as follows

\begin{equation}
\label{pdfOfF2D}
	f_{F2D}(x) = \frac{\lambda \theta x e^{ \frac{\theta}{2}\lambda x^2}}{e^{\frac{\theta}{2}\lambda R^2}-1}.
\end{equation}

The expected value is also derived from Eq. (\ref{pdfOfF2D})

\begin{equation}
\label{ExactEF2D}
	E[X_{F2D}] = \frac{R e^{\frac{\theta}{2} \lambda R^2} }{e^{ \frac{\theta}{2}\lambda R^2}-1} -  \frac{ \int_{0}^{R} e^{\frac{\theta}{2} \lambda x^2} dx} {e^{ \frac{\theta}{2}\lambda R^2}-1},
\end{equation}

which needs the computation of the error function. Chebyshev integral inequality states that

\begin{equation}
	\int_a^b f(x)g(x) dx \geq \frac{1}{b-a} [\int_a^b f(x) dx] [\int_a^b g(x) dx].
\end{equation}

Using this inequality and assuming that 
$f(x)=e^{\frac{\theta}{2} \lambda x^2}$ and $g(x)=2x$, 
both of which are monotonically increasing, we can obtain a 
bound for the second term in $E[X_{F2D}]$ as follows

\begin{equation}
\label{UpperBoundF2DExpectedValue}
	\frac{ \int_{0}^{R} e^{\frac{\theta}{2} \lambda x^2} dx} {e^{ \frac{\theta}{2}\lambda R^2}-1} \leq \frac{2}{R\lambda \theta}.
\end{equation}

The linear function, $g(x)$, is used here for convenience which allows us to integrate the left side of the inequality. In fact, any increasing function can be used, leading to a different inequality. Note that this bound decreases quickly to zero as density increases. Using this inequality, the approximated value of $E[X_{F2D}]$, which 
does not require the computation of error function, is obtained as follows

\begin{equation}
\label{AppEF2D}
	E[X_{F2D}] \approx \frac{R e^{\frac{\theta}{2} \lambda R^2} }{e^{ \frac{\theta}{2}\lambda R^2}-1} -  \frac{2}{R\lambda \theta}.
\end{equation}

\begin{figure}
\centering
\includegraphics[width=0.5\columnwidth]{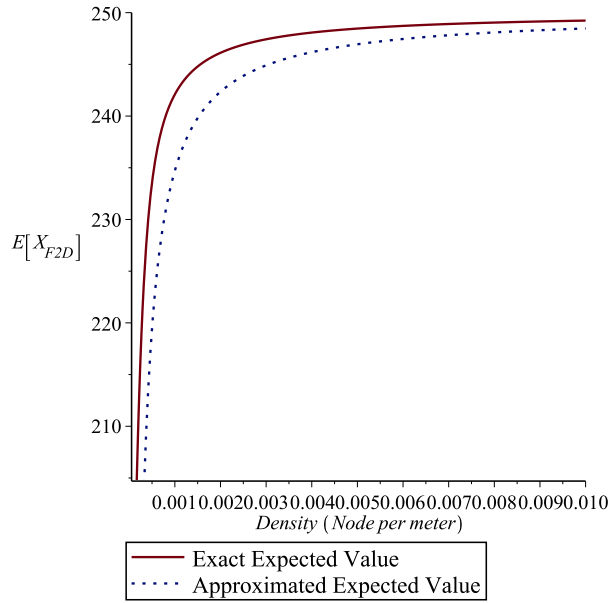}
\caption{$E[X_{F2D}]$ and its approximation.}
\label{ExactAndAppOfEF2D}
\end{figure}

The difference between Eq. (\ref{ExactEF2D}) and (\ref{AppEF2D}) 
is shown in Fig. \ref{ExactAndAppOfEF2D}. It is clear that for realistic node density,
Eq. (\ref{AppEF2D}) accurately approximated the exact expected value.
For the second moment of $X_{F2D}$, we get
\begin{equation}
	E[X^2_{F2D}] = \frac{R^2 e^{\frac{\theta}{2} \lambda R^2} }{e^{ \frac{\theta}{2}\lambda R^2}-1} -  \frac{2}{\lambda \theta} \approx R\times E[X_{F2D}] .
\end{equation}

Hence, Eq. (\ref{eq::linearApproxNF2D}) is a linear approximation of $N_{F2D}$.

\begin{equation}
\label{eq::linearApproxNF2D}
	\bar{N}_{F2D} (x) =  \frac{x}{E[X_{F2D}]} + \frac{E[X^2_{F2D}]}{2E^2[X_{F2D}]} = \frac{1}{E[X_{F2D}]} (x+\frac{R}{2}).
\end{equation}

\section{Simulations}
\label{simulation}
To validate the analytical results obtained in previous sections 
we did an exhaustive network simulation using network simulator 
NS$2$ \cite{ns}. We intentionally changed AODV routing protocol 
to selects random or furthest neighbor based on the scenario. 
In our simulations, all nodes are distributed in a 
line of length $2000m$ randomly. For $2$-D case, all nodes 
were randomly deployed in a region of $2000 \times 1000 m^2$. 
The source and destination are 
located at two different sides of our network. IEEE $802.11$ MAC 
without RTS/CTS mechanism is used in all wireless nodes. The 
distance model with the transmission range of $250m$ is chosen for 
message passing among nodes. 
Interference range and carrier 
sensing range are set as $450m$ and $500m$, respectively.
Nodes are assumed to be stationary and there is 
no other source of traffic in the network. 
Each simulation scenario for each data point is conducted 1000 times, each with different seed value. The confidence interval of 99\%  is also shown in the figures. A single flow consisting of a pair of source and destination nodes located in two far apart edges of the network is simulated for $100$ seconds.

The data rate of the wireless channel is chosen as $1Mbps$. Moreover, nodes are already aware of the maximum throughput 
of a single-hop transmission. 
To find the throughput of a single-hop communication, i.e. $C$, we evaluated the maximum 
throughput by simulating a transmission over $2$ adjacent 
nodes. In our simulation, $C$ is $0.87Mbps$. 
Note that if the value of Short Inter Frame Space (SIFS), Distributed Inter Frame Space (DIFS), MAC header, Acknowledgment (ACK) length or any other protocol specific parameters are available at the transport or network layer protocols, $C$ could be easily obtained mathematically.


Using random neighbor routing, the sending rate of the source node
has been increased to find the rate at which the end-to-end throughput 
reaches its maximum. 
As it is shown in Fig. \ref{figRandomSim}, our 
exact and approximate analysis of IEEE $802.11$ yield almost the same result.
Here, the graph of perfect MAC is 
given just to show how much the imperfection of a MAC layer can affect 
the end-to-end throughput.
Note that the aim of our analysis is to find the maximum achievable 
throughput. Thus, we are not particularly concerned about the values 
our model gives for sending rates greater than the maximum. The reason 
why our model fails to predict values greater than maximum is that our 
model overestimates the probability of collision in this case. When a sending
rate of the source node is above the maximum, the majority of collisions happen 
at the first nodes. Hence, packets that are forwarded successfully out of this 
region can successfully reach the destination. Note that intermediate nodes 
receive fewer packets than the second or third nodes and as a result their 
forwarding rate is less than the maximum throughput. It is worth mentioning that when a source node's sending rate is
less than the maximum end-to-end throughput, almost no collision or contention will occur. As a result, as you can see in Fig. \ref{figRanFurSim}, 
the confidence interval is very close to the average value. However, as the sending rate gets closer to the maximum value, variance over end-to-end throughput increases as the probability of collision and contention increases. 

%

\begin{figure*}[t!]
\begin{center}
 \subfloat[random neighbor routing]{ \includegraphics[width=.4\linewidth]{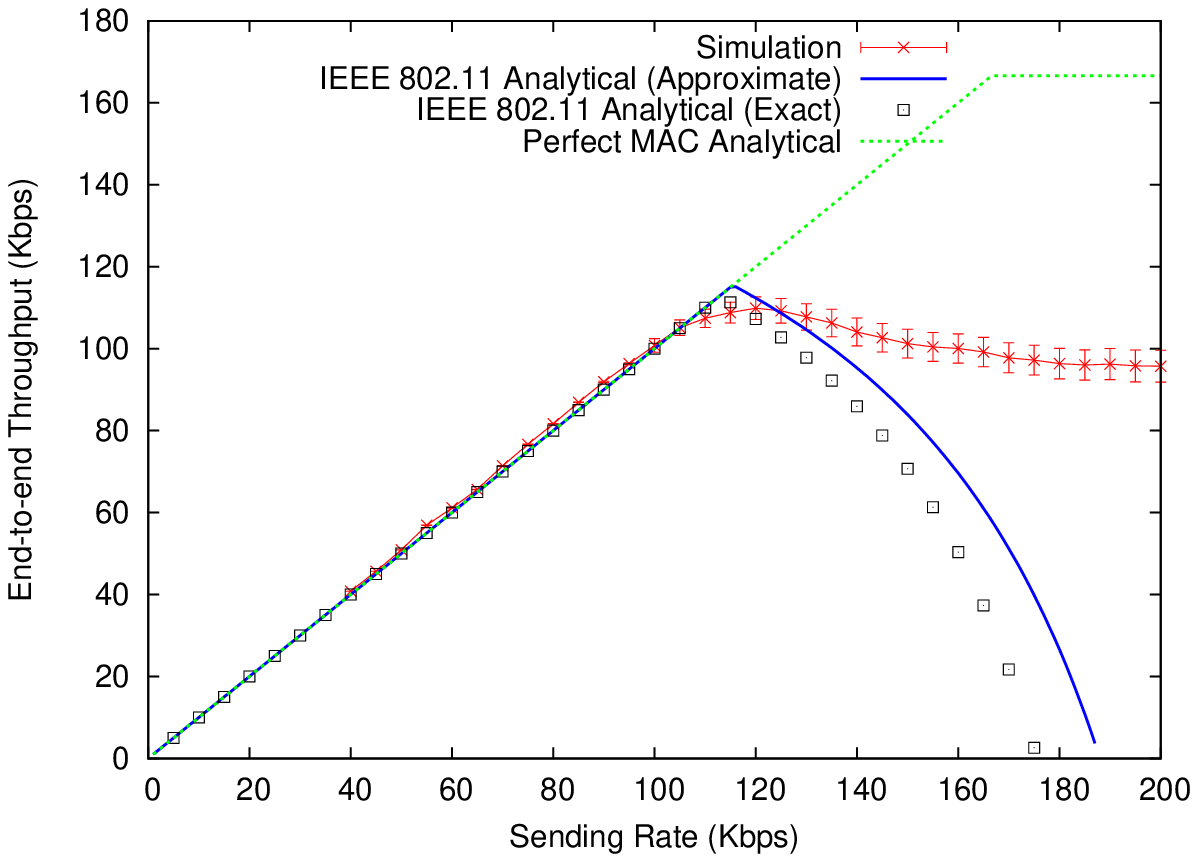}\label{figRandomSim}}
 \subfloat[furthest neighbor routing]{ \includegraphics[width=.4\linewidth]{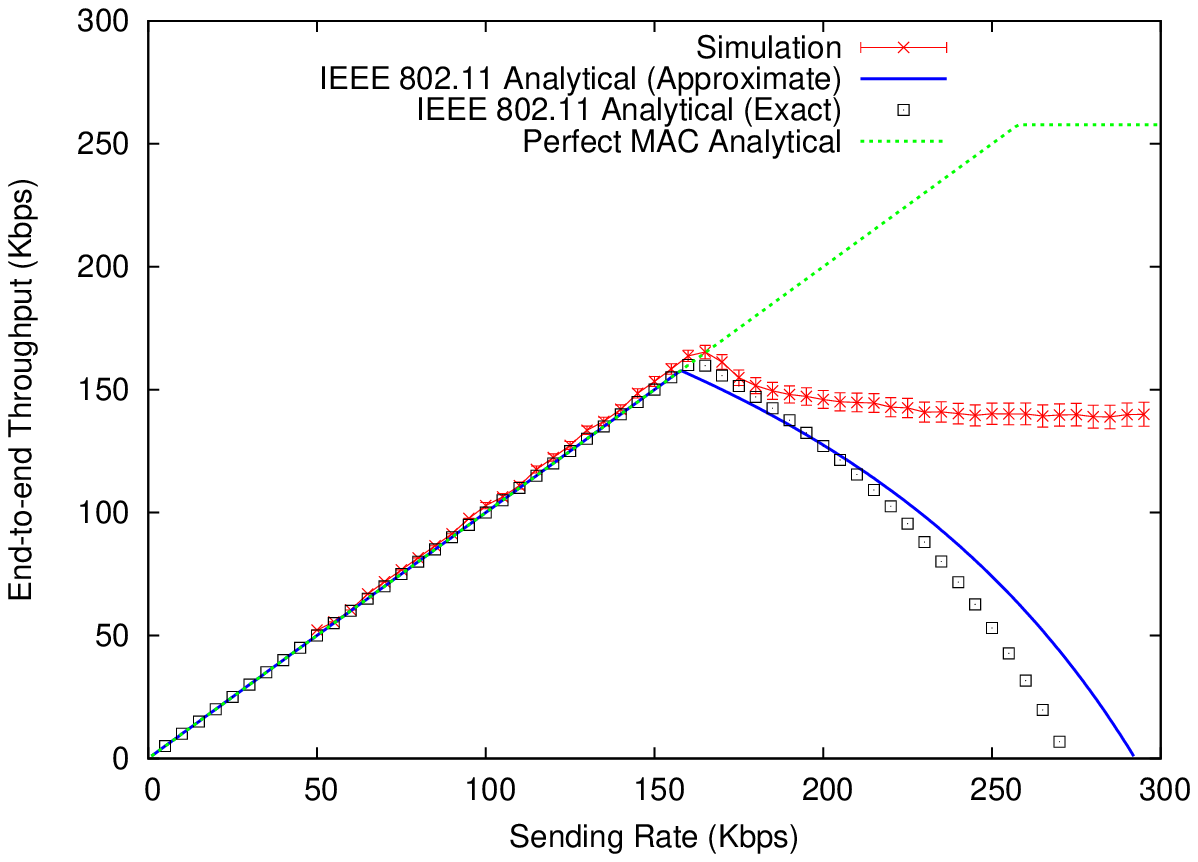}\label{figFurthestSim}}
\caption{End-to-end throughput in $1$-D networks}
\label{figRanFurSim}
\end{center}
\end{figure*}


Fig. \ref{figFurthestSim} illustrates end-to-end throughput for 
node density equals to $0.04$ using furthest neighbor routing. 
This means that each node has 10 neighboring nodes on average. As it is shown in Fig. 
\ref{figFurthestSim}, our exact analysis predicts the maximum 
end-to-end throughput with high precision. Although the approximate 
approach slightly underestimates the maximum end-to-end throughput, 
the difference is negligible and it is still possible to use the approximate 
value in practice.

It is worth noting that in furthest neighbor routing, the number 
of neighbors is of paramount importance. Since the routing policy 
enforces a node to select the furthest neighbor, increasing the 
node density increases the chance of selecting a node closer to 
$R_{tx}$. As a result, the value of $N_F(x)$ decreases as node 
density increases. Hence, we expect an improvement in end-to-end 
throughput as node density increases. As it is shown in Fig. 
\ref{figDensitySim}, the end-to-end throughput increases until 
it reaches a point in which the distance to the next hop is always 
close to $R_{tx}$. 

\begin{figure}
\centering
\includegraphics[width=0.4\columnwidth]{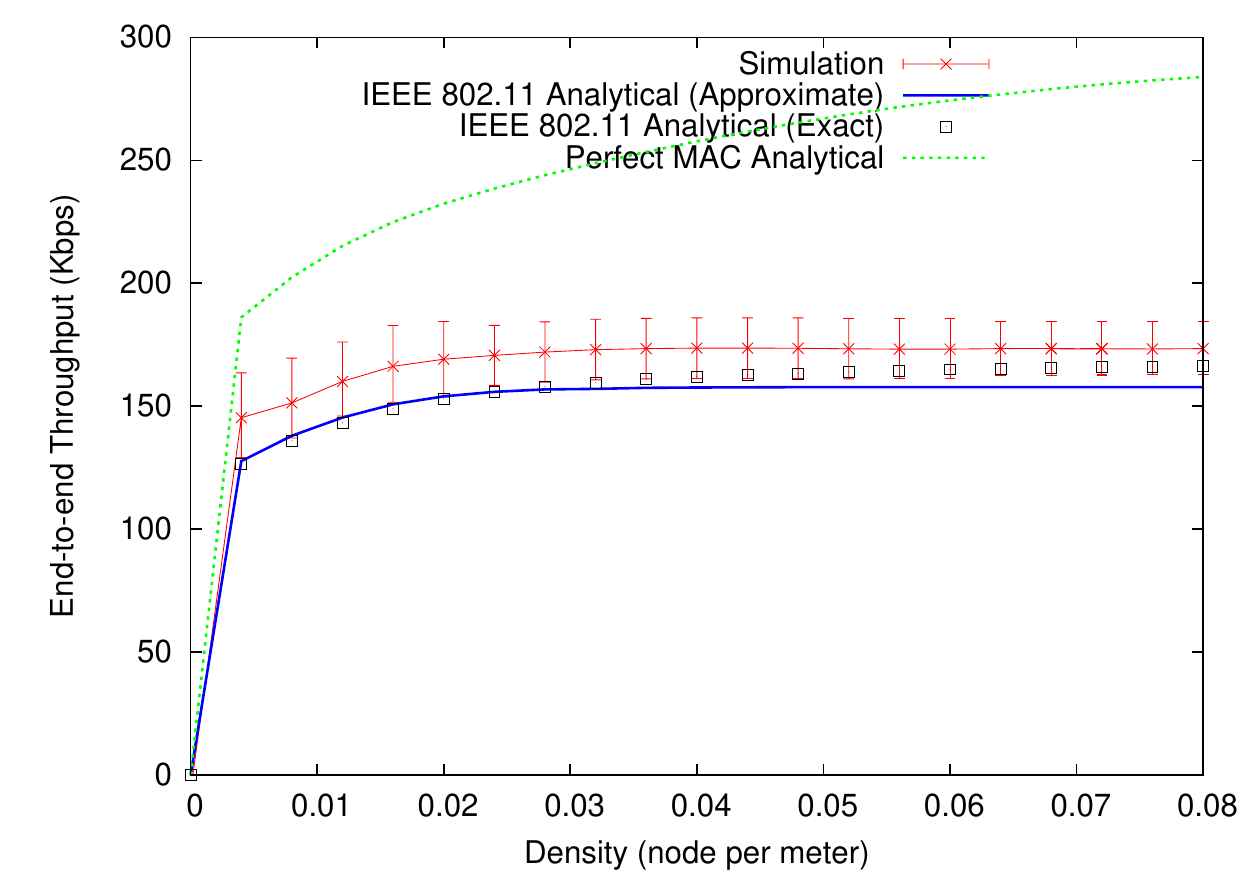}
\caption{Effect of node density on maximum end-to-end throughput in $1$-D networks.}
\label{figDensitySim}
\end{figure}

In Fig. \ref{figAODVSim}, end-to-end throughput of AODV is shown. 
As it is mentioned earlier, the throughput of AODV should be between 
random neighbor routing and furthest neighbor routing. The simulation 
whose results are shown in this figure clearly affirms our argument of 
Section \ref{routingAnalysis}.

\begin{figure}
\centering
\includegraphics[width=0.4\columnwidth]{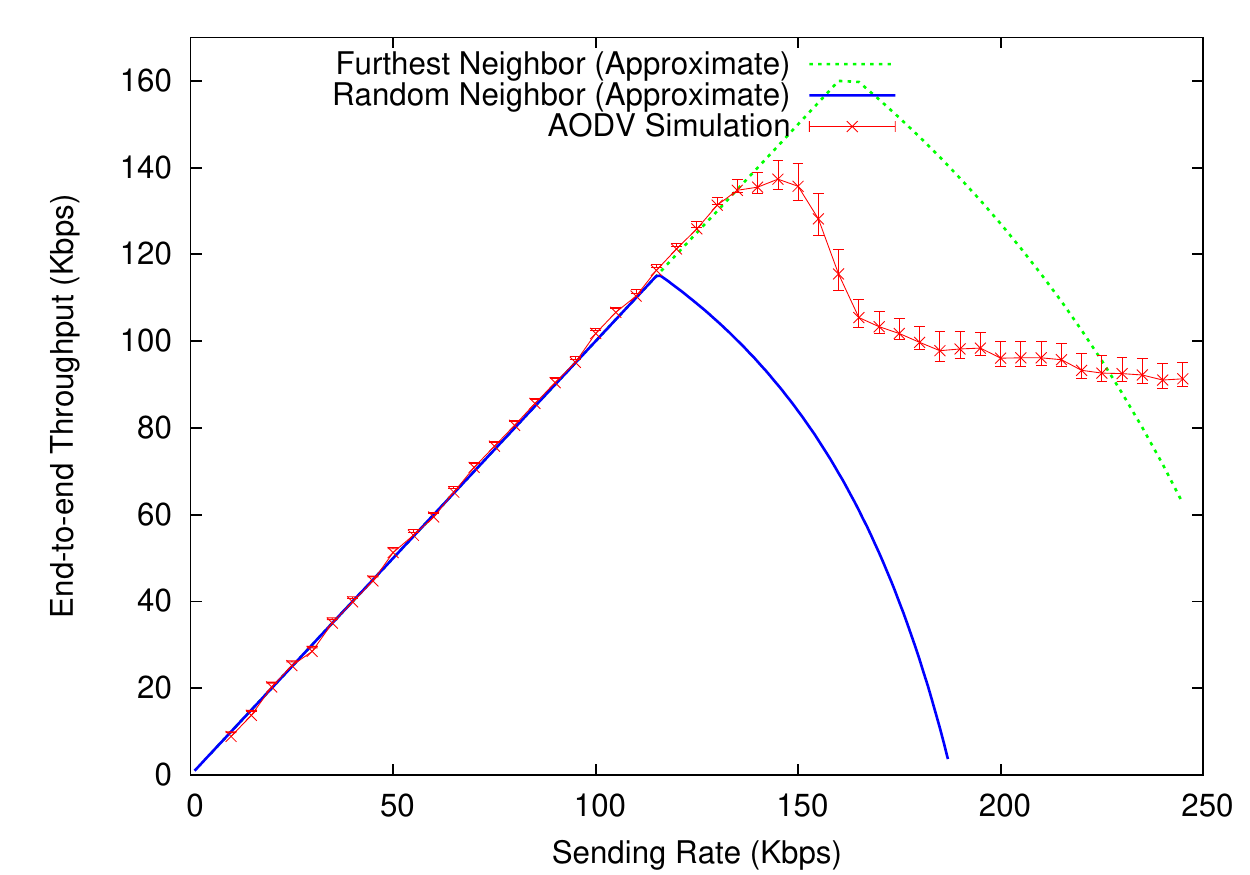}
\caption{AODV routing protocol in comparison with random and furthest neighbor routing.}
\label{figAODVSim}
\end{figure}

\begin{figure}
\centering
\includegraphics[width=0.4\columnwidth]{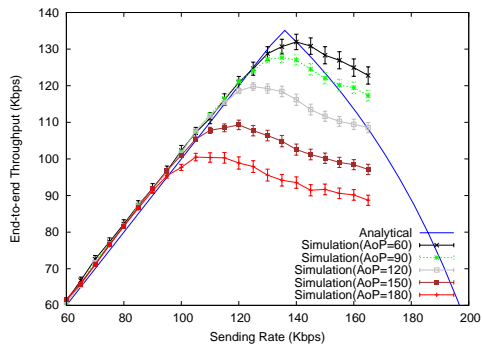}
\caption{End-to-end throughput of random neighbor routing in $2$-D networks ($\lambda$=0.00015).}
\label{SimR2D}
\end{figure}

Fig. \ref{SimR2D} illustrates the end-to-end throughput of random 
neighbor routing in $2$-D networks. Our analytical expression almost 
meets the simulation results when the angle of progression (AoP) is 
small. Note that since the route to the destination usually approaches 
a straight line \cite{straightLine,baccelli2010stochastic}, it is reasonable to assume that the 
AoP is smaller than $60$ degrees in reality. Comparing Fig. \ref{figRandomSim} and \ref{SimR2D}, we have 
noticed that the end-to-end throughput of random neighbor routing 
in $2$-D networks is greater than $1$-D networks. This interesting 
phenomenon stems from the fact that in $2$-D networks random 
neighbor routing tends to select a more distant node than in $1$-D. 
In $1$-D networks, the number of nodes distributed in each small 
region at a distance of $x$ from the source node are statistically 
equal. On the other hand, in $2$-D networks, the number of nodes 
located at a distance of $x$ from the source is increased as $x$ 
gets bigger. Note that nodes at a distance $x$ from the source are 
located at an arc whose length is increased as $x$ gets bigger. For 
this reason, the end-to-end throughput of random neighbor routing 
increases as the dimension of space increases. 
As it is shown in Fig. \ref{SimR2D}, the end-to-end throughput decreases as AoP increases. The reason is that when AoP is smaller the forwarding nodes are more likely to be in a straight line towards the destination. However, when AoP is larger, the number of nodes in the hidden node area becomes more than one. As we discussed in section \ref{$2$-DSection}, in 2-D networks, when the next hops are not chosen in a relative straight line to the destination, the sum of distances 
from the first node to the second node, and the second to the third node is not equal to the distance from the first node to the third one, shown in Fig. \ref{fig2D}. That means that nodes are generally closer together. As a result, contention is higher and the number of nodes in the hidden node area is also greater than one. That is why our model fails to predict the maximum end-to-end throughput when AoP is large.

\begin{figure*}[t!]
\begin{center}
 \subfloat[AoP = $60^{\circ}$]{ \includegraphics[width=.4\linewidth]{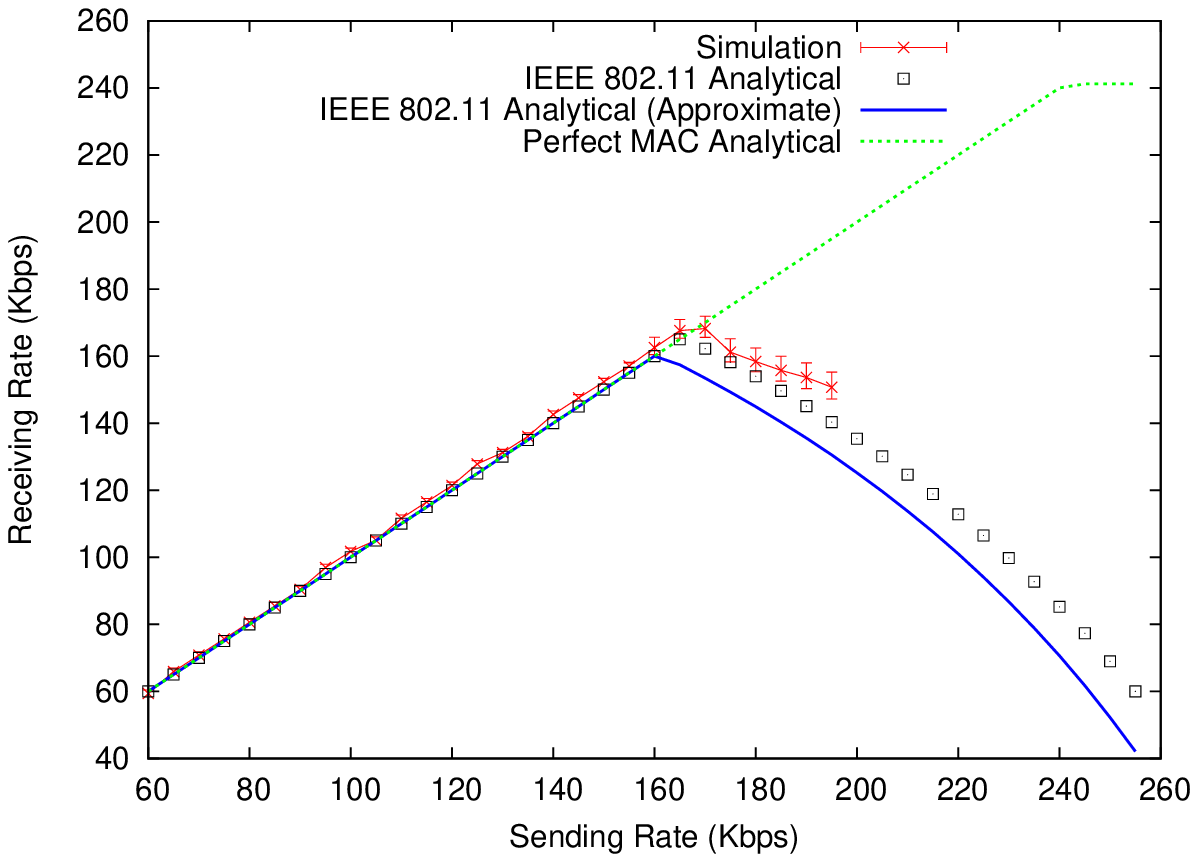}}
 \subfloat[ AoP = $120^{\circ}$]{ \includegraphics[width=.4\linewidth]{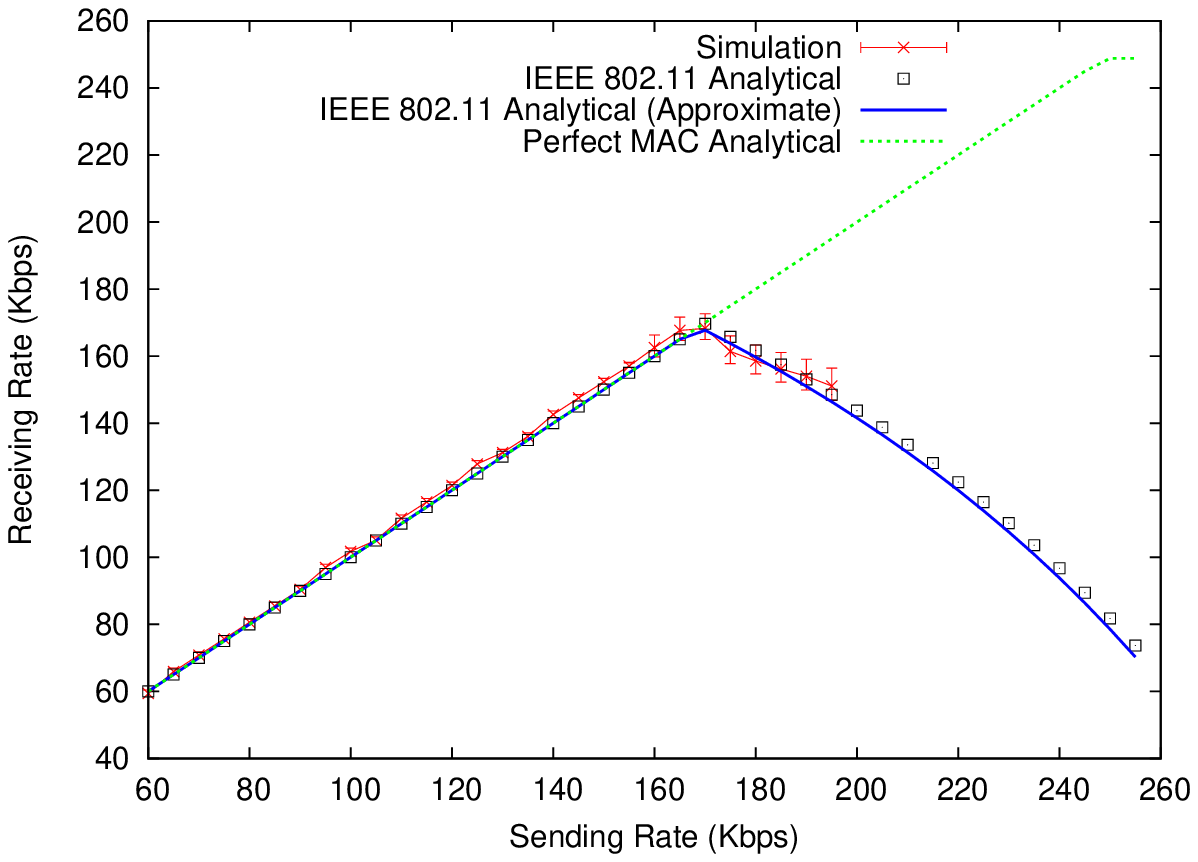}}
\caption{End-to-end throughput of furthest neighbor routing in $2$-D networks ($\lambda$=0.0002).}
\label{figFurthestSim2D}
\end{center}
\end{figure*}

Next, we conducted a series of simulations to validate our analysis of 
furthest neighbor routing, which is plotted in Fig. \ref{figFurthestSim2D}. 
In this figure, \textit{Analytical} and \textit{Perfect MAC Analytical} 
are obtained using Eq. (\ref{ExactEF2D}), and \textit{Analytical (Approximation)} 
is obtained from Eq. (\ref{AppEF2D}). As it is shown, unlike random 
neighbor routing, the angle of progression has a less dramatic effect 
on the maximum throughput. Moreover, as Angle of progression increases, 
the upper bound in Eq. (\ref{UpperBoundF2DExpectedValue}) vanishes to 
zero and consequently the approximation function gets closer to the exact 
function. 

Note that in this figure, the density is $0.0002$, which means that there 
are about $40$ nodes in the transmission range of each node. However, 
the number of nodes in the angle of progression is meaningfully lower. For 
instance, when AoP is $60$ degree, i.e. $\pi /3$, there are only $6.5$ nodes 
available to be selected in the routing process. In fact, we encountered some 
simulations in which there were no path at all between the source and destination. 
Hence, our approximation estimated the maximum value of throughput even for 
reasonable node density.

\begin{figure*}[t!]
\begin{center}
 \subfloat[AoP = $60^{\circ}$]{ \includegraphics[width=.4\linewidth]{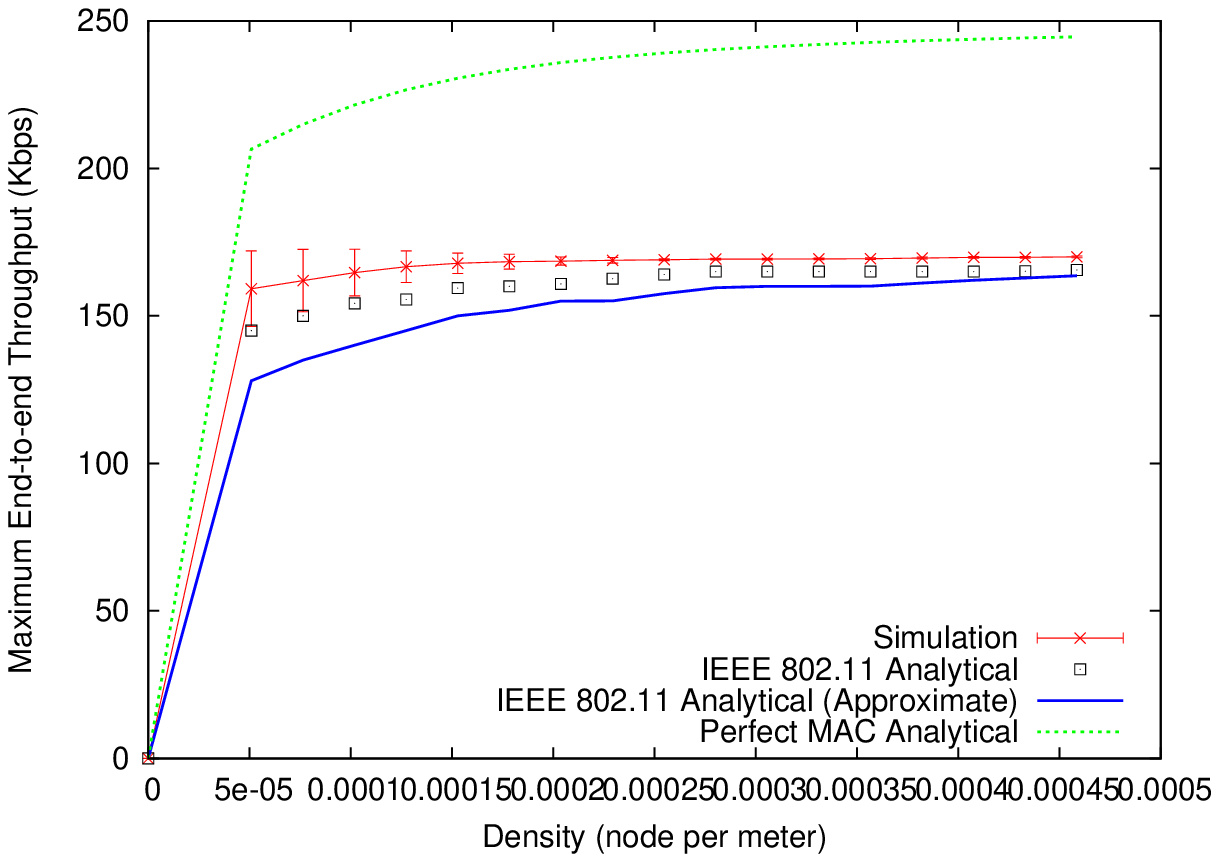}}
 \subfloat[ AoP = $120^{\circ}$]{ \includegraphics[width=.4\linewidth]{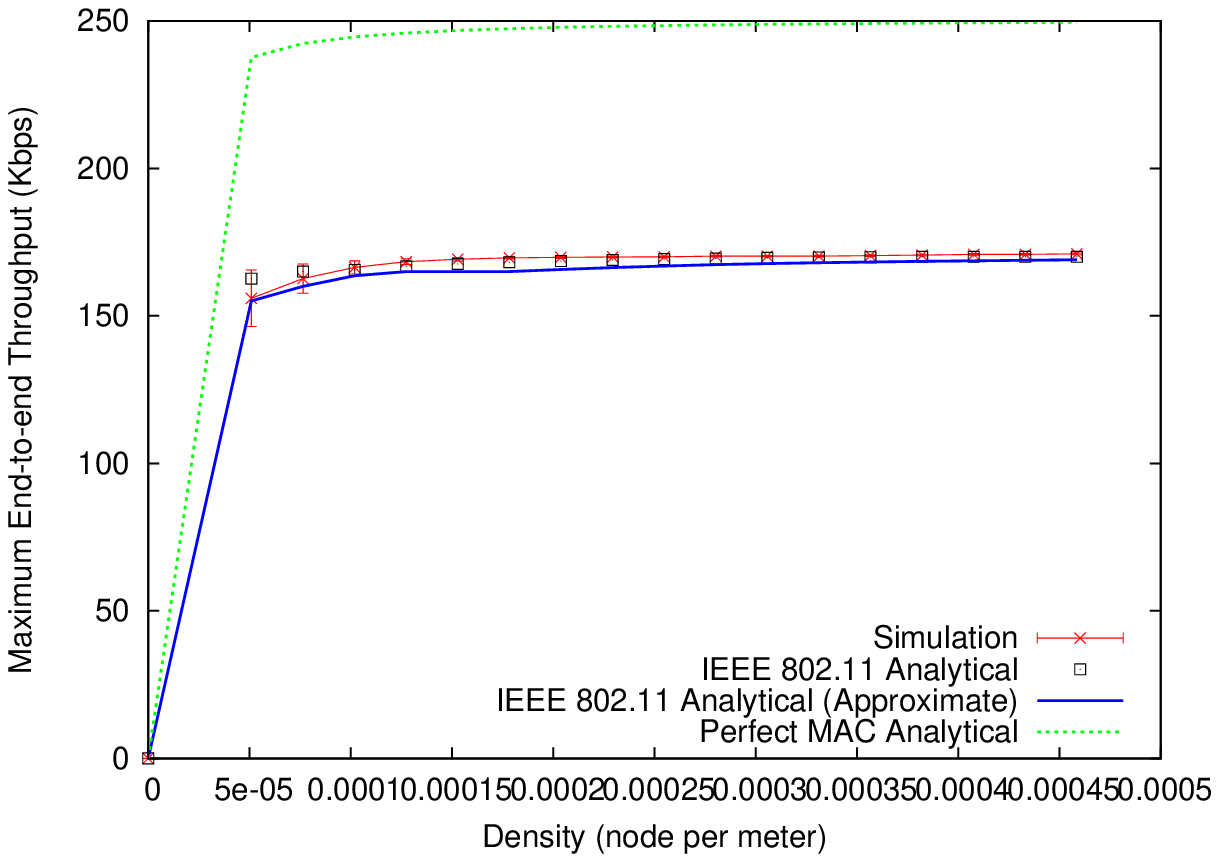}}
\caption{Effect of node density on the maximum end-to-end throughput in $2$-D networks.}
\label{figDensitySim2D}
\end{center}
\end{figure*}

To capture the effect of node density on maximum throughput, 
Fig. \ref{figDensitySim2D} is provided. As node density or AoP 
increases, our analytical formula approaches the simulation 
value. Due to the bound expressed in Eq. (\ref{UpperBoundF2DExpectedValue}), 
it is clear that our approximation underestimates the maximum 
throughput when both density and AoP are low. In most cases, 
however, the difference is negligible.




\section{Conclusion}
\label{conclusion}
This paper has presented a novel analytical approach to obtain 
end-to-end throughput. The routing policy 
along with its effect on the distance between intermediate nodes 
disregarded completely in previous works. In this work, we have  
taken both into consideration to obtain a more precise value for 
maximum end-to-end throughput. We have shown that routing policy 
has an important role in end-to-end throughput which cannot be overlooked. 
It has an influence upon the number of forwarding nodes in crucial regions 
such as transmission range, interference range as well as carrier sensing 
range.
Given a perfect MAC layer in which no hidden or exposed node problem 
occurs, the maximum end-to-end throughput has been obtained for two 
routing policy, including random neighbor routing and furthest neighbor 
routing. Then, the model has been extended to consider imperfection 
of IEEE $802.11$ MAC layer. Due to the relative complexity of the model, 
the approximation method has been also presented so as to render this 
model practical for computationally limited wireless nodes. We have also 
extended our approximation to include $2$-D networks. The validity 
of our analytical approach has been validated by simulation. We believe 
that the proposed model can be used in wireless nodes for admission control 
and flow control. Given the distribution to the next hop for a particular 
routing policy, our methodology also could be used to obtain the maximum 
end-to-end throughput.

As a future work, we will investigate how delay can be obtained in such networks. Due to the shared medium of wireless networks, when a node is transmitting, it affects the communication of neighboring nodes. Classical queue-based approaches cannot deal with this dependency. Therefore, it is more difficult to obtain the end-to-end delay analytically. It is also worth investigating how a different and more realistic channel model can be considered in obtaining end-to-end throughput. Moreover, our analysis can only predict a single flow. When there are multiple flows in the networks whose nodes are close enough to interfere, the end-to-end throughput will be considerably lower. In this scenario, further analysis is needed to find the number of hidden nodes and consider the effect of high contention.

\ifCLASSOPTIONcaptionsoff
  \newpage
\fi

\nocite{*}
\bibliographystyle{IEEEtran}
\bibliography{References}

\end{document}